\documentclass[dvips,aos,preprint]{imsart}

\RequirePackage[OT1]{fontenc}
\RequirePackage{amsthm,amsmath}
\RequirePackage[numbers]{natbib}
\RequirePackage[colorlinks,citecolor=blue,urlcolor=blue]{hyperref}
\RequirePackage{hypernat}

\arxiv{math.PR/0000000}

\usepackage{latexsym,amsmath,amsfonts}
\usepackage{algorithmicx,algpseudocode}
\usepackage{amssymb}
\usepackage{pstricks,pst-node,pst-tree}
\usepackage{thmtools,thm-restate}

\startlocaldefs
\numberwithin{equation}{section}
\theoremstyle{plain}

\newtheorem{theorem}{Theorem}
\newtheorem{proposition}{Proposition}
\newtheorem{corollary}{Corollary}
\newtheorem{lemma}{Lemma}

\newtheorem{assumption}{Assumption}
\newtheorem{definition}{Definition}

\newtheorem{condition}{Condition}
\newtheorem{example}{Example}
\newtheorem{claim}{Claim}

\newtheorem{remark}{Remark}

\def\R{\mathbb{R}} 
\def\Z{\mathbb{Z}} 
\def\N{\mathbb{N}} 
\def\disc{\lambda}
\def\LL{{\rm L}}
\def\ELL{{\rm L}}

\def\RR#1#2{{\left\|#2\right\|_{\LL,#1}}} 
\def\k{{\digamma}}
\def\MM#1#2{{\left\|#2\right\|_{\LL,#1}}} 
\def\Good{{\rm Good}} 
\def\bgamma{{\theta}}

\newcommand{\Ex}[1]{\mathbb{E}\left[#1\right]} 

\newcommand{\Ind}[1]{1_{#1}} 
\newcommand{\Var}[1]{\mathbb{V}\left(#1\right)} 
\renewcommand{\Pr}[1]{\mathbb{P}\left(#1\right)} 

\newcommand{\bigoh}[1]{O\left(#1\right)}
\newcommand{\liloh}[1]{o\left(#1\right)}

\def\sF{\mathcal{F}}

\def\sS{\mathcal{S}}\def\sT{\mathcal{T}}

\def\eps{\epsilon}

\def\bX{{X}}

\def\conf{\widehat{{\sf cr}}}

\def\oconf{\overline{{\sf cr}}}

\def\parent{{\rm par}}
\def\supp{{\rm supp}}

\endlocaldefs

\begin{document}

\begin{frontmatter}
\title{Approximate group context tree\protect\thanksref{T1}}
\runtitle{AGCT}
\thankstext{T1}{First Version: June 2010; Current Version: \today.}

\begin{aug}
\author{\fnms{Alexandre} \snm{Belloni}\ead[label=e1]{abn5@duke.edu}}
\and
\author{\fnms{Roberto I.} \snm{Oliveira\thanksref{t1}}\ead[label=e2]{rimfo@impa.br}}

\thankstext{t1}{Supported by projects {\em Universal} and {\em Produtividade em Pesquisa} from CNPq, Brazil. This work is part of USP project ``Mathematics, computation, language and the brain".}
\runauthor{Belloni and Oliveira}

\affiliation{Duke University and IMPA }

\address{100 Fuqua Drive\\
Durham, NC 27708\\
\printead{e1}\\
}

\address{Estrada Dona Castorina 110\\
Rio de Janeiro, RJ\\
\printead{e2}\\
}
\end{aug}

\begin{abstract}

We study a variable length Markov chain model associated with a group of stationary processes that share the same context tree but each process has potentially different conditional probabilities. We propose a new model selection and estimation method which is computationally efficient. We develop oracle and adaptivity inequalities, as well as model selection properties, that hold under continuity of the transition probabilities and  polynomial $\beta$-mixing. In particular, model misspecification is allowed.

These results are applied to interesting families of processes. For Markov processes, we obtain uniform rate of convergence for the estimation error of transition probabilities as well as perfect model selection results. For chains of infinite order with complete connections, we obtain explicit uniform rates of convergence on the estimation of conditional probabilities, which have an explicit dependence on the processes' continuity rates. Similar guarantees are also derived for renewal processes.

Our results are shown to be applicable to discrete stochastic dynamic programming problems and to dynamic discrete choice models. We also apply our estimator to a linguistic study, based on recent work, by Galves et al \cite{GGGL2012}, of the rhythmic differences between Brazilian and European Portuguese. \end{abstract}

\begin{keyword}[class=AMS]
\kwd[Primary ]{62M05}
\kwd{62M09}
\kwd{62G05}
\kwd[; secondary ]{62P20}\kwd{60J10}
\end{keyword}

\begin{keyword}
\kwd{categorical time series}
\kwd{group context tree}
\kwd{dynamic discrete choice models}
\kwd{dynamic programming}
\kwd{model selection}
\kwd{VLMC}
\end{keyword}

\end{frontmatter}

\section{Introduction}

In this paper, we are interested in applying {\em context tree models}, also known {\em variable length Markov chains} (VLMCs), to the estimation of transition probabilities and dependence structures in discrete-alphabet stochastic processes. Context tree models describe processes where each infinite ``past"~ has a finite suffix -- the {\em context} -- that suffices to determine the transition probabilities. As such, they are generalizations of finite order Markov chains, for which contexts exist and are of fixed length. Context tree processes first appeared in Rissanen's seminal paper \cite{Rissanen1983}, where two appealing traits were noted.
\begin{itemize}
\item {\em Parsimony:} a Markov chain model must have an order parameter that is large enough to distinguish any two pasts with different transition probabilities. By contrast, by using different context lengths for different pasts, one may need less parameters to specify the model. (Incidentally, this motivates the VLMC terminology.)
\item {\em Computationally efficient estimation:} the set of context has a natural suffix tree structure, known as the {\em context tree.} The fact that this is a tree allows for efficiency search over an exponentially large class of models. Rissanen's original Context algorithm for estimating the context tree relied strongly on this.
\end{itemize}
Both traits have continued to play a role over the years as a growing number of papers on context tree models appeared in Statistics \cite{BuhlmannWyner1999,Buhlmann2000,FerrariWyner2003}, Information Theory \cite{WillemsEtAl1995,vert2001adaptive,Garivier2006}, Bioinformatics \cite{Bejerano04} and Linguistics \cite{GGGL2012}. In this last paper, {\em interpretability} of context trees has also played a role, which adds to their interest as practical tools.


In this paper we consider context tree model selection and estimation for a {\em group} of $\LL \geq 1$ stationary processes over a discrete alphabet. These stationary processes have the same context tree but possibly different conditional probability distributions. We refer to this model as \textit{group context tree} alluding to the recent literature on group lasso \cite{YuanLin2006,LouniciPontilTsybakovvandeGeer2009,OWJ2011}. As in the case of group lasso, by combining different processes with similar dependence structure we hope to improve the overall estimation. In addition, the model we consider also allows for processes which are only approximately described by a finite context tree, hence the name {\it approximate group context tree} (AGCT) model.

Although this group context tree setting is new, our estimator and the results we obtain are related to several papers that considered a single stationary process ($\LL = 1$), which we outline briefly. B\"{u}hlmann and Wyner \cite{BuhlmannWyner1999} proved properties of the Context estimator allowing the model to grow with the sample size. They also studied a bootstrap scheme based on fitted VLMCs. Ferrari and Wyner \cite{FerrariWyner2003} consider processes with infinite dependence for which there exist ``good" context tree approximations. They established new results on a sieve methodology based on an adaptation of the Context algorithm. The BIC Context Tree algorithm and its consistency properties have been considered in \cite{CsiszarTalata2006}, \cite{Garivier2006} and \cite{TalataDuncan2009}. Redundancy rates were studied by \cite{CsiszarShields1996} and \cite{Garivier2006}. Several other works contributed to this literature in various directions: see \cite{Buhlmann1999}, \cite{Buhlmann2000}, \cite{vert2001adaptive}, \cite{GarivierLeonardi2010} and the references therein. 

In Section \ref{Sec:GCT}~we propose an estimator for model selection and estimation of conditional probabilities based on context tree models, which {\em does not assume a true VLMC model}. As in Rissanen's original estimator, we first build a full suffix tree for the observed sample, then prune the tree by removing ``statistically insignificant"~nodes. In addition to considering a group of processes, the proposed estimator also differs on how we define insignificance. We use a procedure reminiscent of Lepskii's adaptation method \cite{Lepskii1991}. For each suffix we compute from the sample an (approximate) confidence radius for its vector of transition probability estimates (one for each process). We then recursively prune any leaf node $w$ whose descendants $w'$ in the full sample suffix tree (i.e. the tree prior to pruning) are all ``compatible"~ with the parent of $w$, in the sense that the corresponding confidence regions intersect. By a judicious choice of confidence radii, this procedure automatically balances the variance coming from the random sampling with the bias incurred by the truncation mechanism.


Section \ref{sec:assumptions} details the assumptions we impose on processes, most notably {\em continuity of transition probabilities} (deeper truncation implies arbitrarily good approximation). Based on this, finite sample results on adaptivity and  model selection properties are presented in Section \ref{Sec:Analysis}. In that same section we present stronger results, including oracle inequalities, that require an added assumption of {\em polynomial $\beta$-mixing}. Previous work in the area imposed assumptions that implied a true finite VLMC model, exponential mixing properties and/or non-nullness (positivity) of the transition probabilities, which we manage to avoid here. Moreover, our oracle inequality for the AGCT estimator (Corollary \ref{cor:oracle1}) seems to be the first result of its kind for context tree estimation, even in the single process case.

In Section \ref{Sec:Examples} we present three classes of examples where our general results can be applied. For parametric models (i.e. actual finite-order Markov chains), we derive uniform rate of convergence to transition probabilities, as well as perfect model selection, under weaker assumptions than \cite{BuhlmannWyner1999} (which only covered the single process case). For chains of infinite order with complete connections, we obtain explicit uniform rates of convergence on the estimation of conditional probabilities, which have an explicit dependence on the processes' continuity rates. We also derive explicit uniform rates of convergence for certain renewal processes. In most cases, we show that the group context tree model can lead to improvements on the estimation when compared to the single-process case.

Group context tree models are used in Section \ref{Sec:Applications} to estimate dynamic marginal effects in dynamic choice models (\cite{AguirregabiriaMira2010,BrowningCarro2010}), and to estimate the value function in discrete stochastic dynamic programming problems (\cite{Ross1983,Bertsekas1987,Puterman1994,FariasEtAl2010}). In these applications the objects of main interest are functionals of the conditional probabilities.  We derive uniform bounds on the rate of convergence for the estimates that hold uniformity over all possible contexts and account for model selection mistakes. Furthermore, in Section \ref{Sec:Linguistic} we revisit a study by Galves et al. \cite{GGGL2012}, and apply the AGCT model to understand the difference between the rhythmic features in European and Brazilian Portuguese. A key point is that the AGCT framework allows for the processes to have different transition probabilities.

Section \ref{Sec:Conclusion} discusses variations of the estimator and comparisons, and a final Section adds some further thoughts. Proofs are mostly contained in two Appendices. Simulations and some auxiliary theoretical results are provided in the Supplementary Material \cite{SM-AGCT}.

\subsection{Notation} Let $A$ denote a finite set (called alphabet), and the set of probability distributions over $A$ will be denoted by $\Delta^A$. We use $A^{-1}_{-k}$ to denote all $A$-valued sequences with length $k$, and $A^*= A^{-1}_{-\infty}\cup ( \cup_{k=0}^\infty A^{-1}_{-k})$. The length of a string $w$ is denoted by $|w|$ and, for each $1\leq k\leq |w|$, $w^{-1}_{-k}$ is the suffix of $w$ with length $k$. We also let $w^{-1}_{0}=e$, the empty string.
A subset $\widetilde T\subset A^*$ is a {\em tree} if the empty string $e\in \widetilde T$ and for all $w=w_{-|w|}\dots  w_{-1}\in \widetilde T\backslash \{e\}$ the string $w^{-1}_{-k}=w_{-k}\dots w_{-1} \in \widetilde T$ for any $k\leq |w|$. The {\em parent} of $w$ is denoted by $\parent(w)=w_{-|w|+1}\dots w_{-1}$. An element of a tree $\widetilde T$ that is not the parent of any other element in $\widetilde T$ is said to be a {\em leaf} of $\widetilde T$. For $w,w'\in A^*$, we write $w\preceq w'$ if $w$ is a suffix of $w'$.

We associate with each tree $\widetilde T$ and each $x=\dots x_{-3}x_{-2}x_{-1}\in A^{-1}_{-\infty}$ a suffix $\widetilde T(x)$ of $x$ with the following rule:
\begin{itemize}
\item If $x^{-1}_{-k}\in \widetilde T$ for all $k\in\N$, then $\widetilde T(x)=x$;
\item Otherwise, take the largest $k\in\N$ with $x^{-1}_{-k}\in \widetilde T$ and set $\widetilde T(x)=x^{-1}_{-k}$. (Note that this is the empty string if $k=0$.)\end{itemize}

The strings of the form $\widetilde T(x)$ where $x$ ranges over $A^{-1}_{-\infty}$ will be called the {\em terminal nodes} of $\widetilde T$. Notice that all terminal nodes are either leaves or infinite strings. For two sequences $a_n$, $b_n$ we denote  $a_n\lesssim b_n$ if $a_n=O(b_n)$. The indicator function of an event $E$ is denoted by $\Ind{E}$, and for $q\geq 1$ the $\|\cdot\|_{\LL,q}$-norm of a vector $v \in \mathbb{R}^\LL$ is defined as
\begin{equation}\label{DefLq}\|v\|_{\LL,q} = \left( \frac{1}{\LL}\sum_{\ell=1}^\LL |v_\ell|^q \right)^{1/q}.\end{equation}

\section{Setting for Group Context trees}\label{Sec:GCT}

A pair $(\widetilde T,\widetilde p)$ will correspond to a tree $\widetilde T$ and a mapping $\widetilde p$ that assigns to each terminal node $v$ of $\widetilde T$ a probability distribution $\widetilde p(\cdot | v)$ over a finite alphabet $A$. A stationary ergodic process $\bX\equiv (X_k)_{k\in\Z}$ will be said to be {\em compatible with} $(\widetilde T,\widetilde p)$ if:
$$\Pr{X_0=a\mid X^{-1}_{-\infty}} = \widetilde p(a\mid \widetilde T(X^{-1}_{-\infty}))\mbox{ almost surely.}$$


On a group context tree model we have a family $\bX=(\bX(\ell))_{\ell=1}^{\LL}$ of $\LL$ independent and stationary processes
$$\bX(\ell)\equiv (X_k(\ell))_{k\in\Z}\,\,(1\leq \ell\leq \LL),$$
a single context tree $T^*$, and (possibly distinct) probability distributions $p_\ell$, $\ell = 1,\ldots,\LL$, such that the $\ell$th process is compatible with
$(T^*,p_\ell)$ for $\ell = 1,\ldots,\LL$.
Note that $T^*$ is possibly infinite so that this is not a restriction/assumption on the model. Moreover, if the $\ell$th process is compatible with a context tree $T^\ell$, we have $T^* = \cup_{\ell=1}^\LL T^\ell$, and we may redefine $p_\ell$ correspondingly.


To quantify the approximation error we use a metric $d_\ell: \Delta^A \times \Delta^A\to
[0,1]$ for each process, $\ell=1,\ldots,\LL$ and write the associated $\LL$-vector $d(p,q)=(d_1(p_1,q_1),\ldots,d_\LL(p_\LL,q_\LL))'.$ We will aggregate the approximation errors across processes through $\MM{\k}{d(p,q)}$ where $\MM{\k}{\cdot}$ is the norm defined in (\ref{DefLq}). For simplicity, we consider all metrics $d_\ell$ to be equal and of a certain specific kind. Namely, there exists a collection $\sS$  of subsets of $A$ such that:
\begin{equation}\label{def:metrics}d_\ell(p_\ell,q_\ell) =\sup_{S\in\sS}|p_\ell(S)-q_\ell(S)|, \ \ \ell=1,\ldots,\LL.\end{equation}
Our main interests are in the $\ell_1$ metric, where $\sS=2^A$ consists of all subsets of $A$, and the $\ell_{\infty}$ metric, where $\sS$ consists of all singletons of $A$. 

\subsection{The AGCT estimator}

In this section we discuss the model selection method which leads to the estimation of the conditional probabilities from a sample of $\LL$ processes. For each $\ell = 1,\ldots,\LL$, our sample consists of a string of size $n$ with symbols from $A$ denoted as $X_1^n(\ell)\equiv (X_{1}(\ell),\dots,X_{n}(\ell))$. For a string $w\in A^*$, we let $N_{k,\ell}(w)$ denote the  number of occurrences of $w$ in $X_1^k(\ell)$. \footnote{Formally defined for $k\geq |w|+1$, so that $N_{k,\ell}(w)$ denotes the number of indices $i$, $|w|\leq i \leq k$, with $X_{i-|w|}^i(\ell)=w$.} (For notational convenience we assume that the length $n$ of the sample of each process is the same but the analysis does not rely on that.)

The algorithm proceeds in three steps: Initialization, Identification of Removable Nodes, and Pruning. Next we describe in detail the procedure. In what follows we let $E_n$ be the suffix tree that contains every string $w\in A^*$ which appears in all $\LL$ data sequences of the sample $X_1^{n-1}$, namely \begin{equation}\label{def:En} E_n = \left\{ w\in A^* \ : \  \min_{\ell=1,\ldots,\LL}N_{n-1,\ell}(w)>0 \right\}.\end{equation}



\noindent {\bf Step 1: Initialization.} {\em For each $w\in E_n$  we specify a conditional probability estimate and a confidence radius:
$$\hat{p}_{n,\ell}(a|w)\equiv\frac{N_{n,\ell}(wa)}{N_{n-1,\ell}(w)} \ \ \mbox{and} \ \ \conf_\ell(w) \mbox{(to be specified)}, \ \ \mbox{for} \ \ a \in A, \  \ell=1,\ldots,\LL.$$}
The estimator $\hat{p}_{n,\ell}(a|w)$ is a nonparametric estimate for the transition probability $p_\ell(a|w)$. The confidence radius $\conf_{\ell}(w)$, to be specified in Section \ref{sec:datadriven1} below, depends on the choice of $\k$. With high probability, it is essentially an upper bound for the distance between $p(\cdot|w)$ and $\hat{p}_{n}(\cdot|w)$, up to a bias factor that comes from truncating the past of the process at $w$ (this is related to the {\em continuity rates}, cf. Assumption \ref{assum:cont} below).\\

\noindent {\bf Step 2: Identifying Removable Nodes.} {\em For a fixed constant $c > 1$, define for each $w \in E_n$: {\small \begin{equation}\label{Def:IsPrunable}{\sf
CanRmv}(w)\equiv \left\{\begin{array}{ll} 1, & \mbox{if for all }\
w', w''\in E_n\mbox{ with }w\preceq w',  \parent(w)\preceq w''\\ & \MM{\k}{
d(\hat{p}_n(\cdot|w'),\hat{p}_n(\cdot|w''))}\leq c\RR{r}{\conf(w')}+c\RR{r}{\conf(w'')};\\ 0, &
\mbox{otherwise.}\end{array} \right.\end{equation}}}

Intuitively, ${\sf CanRmv}(w)=1$ means that we can remove $w$, which happens if and only if, for any two nodes $w\preceq w'$, $\parent(w)\preceq w''$, the distance between the corresponding transition probability estimates is smaller than the sum of the noise levels at the nodes. The slack factor $c>1$ allows us to keep a check on the bias that might be incurred by removing $w$. Our analysis in the Appendix shows that using $c>1$ implies that, with high probability, this bias will not be much larger than the noise\footnote{See the proof of Lemma \ref{lem:adapt}}. This is similar e.g. to the slack parameter used in \cite{BickelRitovTsybakov2009}, and we recommend $c=1.01$ in practice.\\


\noindent {\bf Step 3: Pruning.} {\em Let $\widehat{T}_n\leftarrow E_n$. Prune any leaf of $\widehat T_n$ with {${\sf
CanRmv}(w)=1$}. Repeat until all leaves of $\widehat T_n$ have {${\sf
CanRmv}(w)=0$}. Return $(\hat P_n, \widehat T_n)$ where $$\hat P_n(a|x) \equiv \hat p_n(a|\widehat T_n(x))$$ for all $x \in A^{-1}_{-\infty}$ and $a \in A$.}\\

This last step keeps the smallest subtree of $E_n$ containing all nodes that {\em cannot} be removed (i.e., for all $w \in \widehat T_n$ we have {\sf CanRmv}$(w)=0$). For completeness we provide detail algorithm in Figure \ref{fig:prune} in Appendix \ref{Sec:Algorithm} of the Supplementary Material  \cite{SM-AGCT}.  The context tree $\widehat T_n$ is our selected model, and the transition probability estimate $\hat P_n$ is compatible with it. We will show that pruning typically removes high-noise nodes, and the bias incurred by pruning is kept manageable by the test in {\sf CanRmv}.

\subsection{Data-driven choices of confidence radii}\label{sec:datadriven1}

The performance of our algorithm is heavily dependent on choices of confidence radii $\conf_\ell(w)$. As noted above, we will choose those so as to bound from above the deviations $\|d(\hat p_{n}(\cdot|w),p(\cdot|w))\|_{\LL,\k}$ up to an extra error term depending on the continuity rates. There is an important tradeoff between large confidence radius that introduce large bias and small confidence radius that do not properly account for the noise in the data. In this section, we present choices that achieve good balance between these factors. These choices ultimately derive from the self-normalized martingale inequalities that we present in the Appendix and Supplementary Material.

\begin{definition}[First choice of confidence radius]\label{def:firstchoice}Let $1-\delta$, $\delta\in (0,1)$ be our desired confidence level. For $w\in E_n$, $\ell=1,\ldots,\LL$, let  $$\conf_\ell(w)\equiv\sqrt{\frac{4}{N_{n-1,\ell}(w)}\,\left(2\ln(2 + \log_2N_{n-1,\ell}(w)) + \ln\left(\frac{n^2\,\ELL\,|\sS|}{\delta}\right)\right)}.$$
\end{definition}

The choice above satisfies $\conf_{\ell}(w) \sim \sqrt{\log (n\LL/\delta)/N_{n-1,\ell}(w)}$. This choice exhibits the same behavior as in the case of a single group ($\LL=1$) provided  $\log \LL \lesssim \log n$, which encompasses most cases of interest. The choice in Definition \ref{def:firstchoice} is desired when we want our estimates of the transition probabilities to be uniformly good approximations. The next proposal for confidence radius is appropriate when the number of processes is large and we want our estimates to be good on average.

\begin{definition}[Second choice of confidence radius]\label{def:secondchoice} Let $1-\delta$, $\delta\in (0,1)$, be the desired confidence level. Assume the condition:
\begin{equation}\label{Def:star}\,\LL\geq 6\ln\left(\frac{n^2}{\delta}\right)\end{equation} and for $w\in E_n$, $\ell=1,\ldots,\LL$, let
{\small $$\conf_\ell(w)\equiv\sqrt{\frac{4}{N_{n-1,\ell}(w)}\,\left(2\ln(2 + \log_2N_{n-1,\ell}(w)) + \ln|\sS| + 1 +\sqrt{\frac{6\ln\left(\frac{n^2}{\delta}\right)}{\ELL}}\right)}.$$}
\end{definition}

In this case, because of (\ref{Def:star}), the rate of $\conf_{\ell}(w)$ is $\sqrt{\log \log n/N_{n-1,\ell}(w)}$ improving upon the single-process case. This is remarkably close to the error in the estimation of probabilities if the model was known in advance.

\section{Assumptions}\label{sec:assumptions}

In this section we state the main assumptions on the processes $(X(\ell))_{\ell=1}^{\LL}$ for our main results. For clarity, we decided to use relatively transparent hypotheses, but slightly more general assumptions can be imposed with very few changes.

\subsection{Basic distributional assumptions}

We start with the simplest assumptions that allow for effective use of the group structure, in that we consider the same ``prefixes" for all processes. To make this precise, we define the {\em support} ${\rm supp}_\ell$ of process $X(\ell)$ as the set:
\begin{equation*}\supp_\ell\equiv \{x^{-1}_{-\infty}\in A^{-1}_{-\infty}\,:\,\forall k\in\N,\, \Pr{X^{-1}_{-k}(\ell)=x^{-1}_{-k}}>0\},\end{equation*}
and formally state our condition.

\begin{assumption}[Framework]\label{assum:supp}We have $\LL$ processes
$$X(\ell) = (X_k(\ell))_{k\in \Z},\,1\leq \ell\leq \LL$$taking values in the same discrete alphabet $A$ which are independent and stationary. All processes have the same (potentially infinite) context tree $T^*$ and (potentially different) transition probabilities $p_1,\dots,p_\LL$. The sets ${\rm supp}_{\ell}$, $1\leq \ell\leq \LL$, are all equal. We denote by ${\rm supp}\equiv {\rm supp}_1$. We observe $\{X_1^n(\ell)\}_{\ell=1}^{\ELL}$, samples of length $n\geq 9$ of the stochastic processes $\{X(\ell)\}_{\ell=1}^{\LL}$.\end{assumption}

\subsection{Continuity rates and Mixing}

The uniform control we aim for essentially requires that truncating the past of the process at some past time $-k$, $k\gg 1$, is not too hurtful for the transition probabilities.

\begin{assumption}[Continuity]\label{assum:cont}The processes $X(\ell)$, $1\leq \ell\leq \LL$, are {\em continuous}. That is, for each $\ell$, there exists a version of the conditional probabilities $p_\ell$ of the $X(\ell)$ process such that the quantities:
$$\gamma_\ell(x^{-1}_{-k})\equiv \sup\limits_{y,z\in A^{-1}_{-\infty}\,:\,y^{-1}_{-k}=z^{-1}_{-k}=x^{-1}_{-k}}d_\ell(p_\ell(\cdot|y),p_\ell(\cdot|z))$$
converge to $0$ as $k\to +\infty$, for all $x^{-1}_{-\infty}$, where $d_\ell$ is a metric as in (\ref{def:metrics}).\end{assumption}
The numbers $\gamma_\ell(\cdot)$ are the continuity rates of process $\ell$. A compactness argument implies that their convergence to $0$ is {\em uniform} in $x \in A^{-1}_{-\infty}$. However, our estimator will {\em adapt} to the continuity rates, meaning that it will tend to do better on pasts that are ``more continuous."

\section{Finite Sample Analysis}\label{Sec:Analysis}

In this section we derive our main theoretical results on the performance of the estimates proposed in Section \ref{Sec:GCT}. 

\subsection{Main results: adaptivity and an oracle inequality}

We can now state our main result.

\begin{theorem}[Main theorem; proven in Appendix \ref{App:ProofsSec4}]\label{thm:main}Under Assumptions \ref{assum:supp} and \ref{assum:cont}, let $\widehat{T}_n$ and $\widehat{P}_n$ denote the tree and transition probabilities output by the AGCT algorithm with $\delta \in (0,1)$, $c>1$ and one of the options below:
\begin{itemize}
\item{\em General case.} We use any $\k\in [1,\infty]$, take $r=\k$ and use the confidence radii as in Definition \ref{def:firstchoice}.
\item{\em Many processes.} In this case we assume condition (\ref{Def:star}) in Definition \ref{def:secondchoice}, take $\k=1$, $r=2$ and use the confidence radii in that definition.\end{itemize}
Then, the following facts hold simultaneously with probability at least $1-\delta$
\begin{enumerate}
\item The estimated tree is contained in the correct tree: $\widehat{T}_n\subset T^*$.
\item Uniformly over $x\in {\rm supp}$, we have
\end{enumerate}
\vspace{-0.5cm}\begin{eqnarray*}\|d(p(\cdot|x),\widehat{P}_n(\cdot|x))\|_{\LL,\k} \leq &{\displaystyle \inf_{T}}\, \frac{2c+2}{c-1}\|\gamma(T(x))\|_{\LL,\k} + (1+2c)\,\|\conf(T(x))\|_{\LL,r}.\end{eqnarray*}
\end{theorem}

Theorem \ref{thm:main} contains two assertions that hold with high probability. Firstly, the AGCT estimator does not give a bigger tree than necessary: this is advantageous when there is a true, finite VLMC model with a small $T^*$. However, note that, in general, $T^*$ might contain some infinite paths.

Secondly, Theorem \ref{thm:main} shows that our estimator adapts to the continuity rates of the process in a very strong, pastwise sense. The transition probabilities for more frequent pasts are better approximated because the confidence radii $\conf_\ell(T(x))$ decrease when the $N_{n-1,\ell}(T(x))$ increase. This is enough to imply the almost sure converge of the AGCT probability estimates to the transition probabilities for continuous, ergodic processes, when the sample size $n$ increases and the values of $\delta=\delta^{(n)}$ chosen are summable.

An added feature is that, under (\ref{Def:star}), we may use the second choice of confidence radii in Definition \ref{def:secondchoice} (with $\k=1$, $r=2$) and obtain faster rate of convergence by a $\sqrt{ \log n / \log\log n}$ factor relative to the choice in Definition \ref{def:firstchoice}. This is indeed the case for some processes studied in more detail in the Supplementary Material \cite{SM-AGCT}.

\begin{remark}[Generality of Adaptivity] The result in Theorem \ref{thm:main} holds for any stationary process. The generality of Theorem \ref{thm:main} is achieved through the use of self-normalized martingale inequalities derived in the Supplementary Material \cite{SM-AGCT}. Those inequalities are used to establish the validity of the data-driven choice of the confidence radius. However, the rates of convergence depend on sample realization through the confidence radius. In order to derive explicit rates of convergence, it is necessary to control how fast the $\LL$ processes lose memory, see Section \ref{Sec:Beta}.
\end{remark}

\subsection{Main results for $\beta$-Mixing Processes}\label{Sec:Beta}

In this section we assume the processes satisfy a polynomial $\beta$-mixing condition, which is known to hold for a wide class of processes. (This property can sometimes be derived from the continuity rates, see Section \ref{Sec:Examples}.)
Recall that a process $X^{+\infty}_{-\infty}$ with values in a finite alphabet $A$ is said to be $\beta$-mixing (or absolutely regular) if there exists a function $\beta:\N\to[0,1]$ with $\lim_{b\in \N, b\to \infty}\beta(b)=0$ and $\forall k\in \Z,\,s\in\N$:
$$\beta(b)\geq \Ex{\sup\limits_{E\subset A^s}\left|\Pr{X_{k+b}^{k+b+s-1}\in E\mid X^{k}_{\infty}} - \Pr{X_{k+b}^{k+b+s-1}\in E}\right|}.$$
The function $\beta(\cdot)$ is called a ($\beta$-)mixing rate function for $X^{+\infty}_{-\infty}$. We assume:

\begin{assumption}[Polynomial $\beta$-mixing]\label{assum:polybeta}The $\LL$ processes $X(1),\dots,X(\LL)$ are all polynomially $\beta$-mixing with common rate function $\beta(b)\equiv \Gamma\,b^{-\bgamma}$ ($b\in\N$), where $\Gamma,\bgamma>0$.\end{assumption}

This extra assumption will allow us to control how ``typical" context trees behave as estimators, which in turn allows us to establish guarantees for the proposed AGCT estimator. To characterize the set of typical trees, recall that under Assumption \ref{assum:supp} the processes $X(1),\dots,X(\LL)$ have the same support ${\rm supp}$, and we define:
\begin{equation}\label{eq:defpiell}\pi_\ell(w)\equiv \Pr{X^{-1}_{-|w|}(\ell) = w}\,\,(w\in {\rm supp},1\leq \ell\leq \LL).\end{equation}
For a finite tree $T$, define $\pi_T$ as the minimum stationary probability of a leaf node,
$$\pi_T:= \min\{\pi_\ell(w)\,:\, 1\leq \ell\leq \LL,\, w\in {\rm supp}  \mbox{ is a leaf of }T\},$$
and let $h_T$ denote the height of $T$,
$$h_T:= \max\{|w|\,:\, w\in{\rm supp}\mbox{ is a leaf of }T\}.$$

\begin{definition}[Typical trees]\label{def:typical} For $(h,\pi_*)$, define the set of {\em typical} trees $\sT(h,\pi_*)$ as the set of all finite trees $T$ satisfying $\pi_T\geq \pi_*$ and $h_T\leq h$.\end{definition}

Define also the population analogues of confidence radii
$$\oconf_{\ell}(w)\equiv \left\{\begin{array}{l} \sqrt{\frac{8}{\pi_\ell(w)\,n}}\sqrt{2\ln(2 + \log_2 \{\pi_\ell(w)\,n/2\})+ \ln\left(\frac{|\sS|\LL n^2}{\delta}\right)}  \mbox{, or } \\
\sqrt{\frac{8}{\pi_\ell(w)\,n}}\sqrt{2\ln(2 + \log_2\{\pi_\ell(w)\,n/2\}) + \ln|\sS| + 1 +\sqrt{\frac{6}{\ELL}\ln\left(\frac{n^2}{\delta}\right)}}\\
\end{array}\right.$$ where $N_{n-1,\ell}(w)$ is replaced by $\pi_\ell(w)n/2$ in $\conf_\ell(w)$.

The next result exploits the $\beta$-mixing condition to provide finite sample bounds that depend on the population confidence radii of typical trees.

\begin{theorem}[Adaptivity for $\beta$-mixing; proven in Appendix \ref{Ap:betamixingtypicality}]\label{thm:adapt} Make Assumption \ref{assum:polybeta}~in addition to the assumptions of Theorem \ref{thm:main}, and consider the typical trees $\sT(h,\pi_*)$ with parameters $h\in \N, \pi_*>0$ such that for $\delta_0\in(0,1/e)$ \begin{equation}\label{Condition:n}
n\geq 2\,\max\left\{40h,\left\lceil\frac{48\,\Gamma\,\LL}{\pi_*\,\delta_0}\right\rceil^{1/\bgamma}\right\}\times \left\{1 + \frac{1200}{\pi_*}\,\log\left(\frac{24\,(h+1)}{\delta_0\,\pi_*}\right)\right\}.\end{equation}
Then, the following inequality holds with probability at least $1-\delta-\delta_0$, simultaneously over all $x \in {\rm supp}$:
\begin{equation*}\|d(\widehat{P}_n(\cdot|x),p(\cdot|x))\|_{\LL,\k} \leq \inf_{T\in \sT(h,\pi_*)}\mbox{$\frac{2c+2}{c-1}$}\|\gamma(T(x))\|_{\LL,\k} + (1+2c)\,\|\oconf(T(x))\|_{\LL,r}.\end{equation*}
\end{theorem}

Theorem \ref{thm:adapt} shows that the estimator balances continuity rates and population confidence radii over the set of typical trees. The parameters $\pi_T^{-1}$ and $h_T$ of these trees may grow polynomially with the sample size $n$, and this allows for the use of very deep nodes for the estimation of difficult pasts. This strong adaptivity property may be rephrased as an {\em oracle inequality} when $\k=\infty$.

\begin{corollary}[Oracle inequality]\label{cor:oracle1} In the setting of Theorem \ref{thm:adapt}, take any choice $(h,\pi_*)$ that satisfies (\ref{Condition:n}) and set $\k=\infty$, $\delta=n^{-a}$ with $a>0$. Then there exists a constant $C>0$ depending only on the slack parameter $c>1$, on the alphabet $|A|$ and on the exponent $a>0$ of $\delta$, such that with probability at least $1-n^{-a}-\delta_0$,
\begin{eqnarray*}\sup\limits_{\stackrel{\stackrel{{T}\in\sT(h,\pi_*)}{\tilde{p}\text{ compatible}}}{\text{ with }{T}}}\,\left(\sup_{x\in {\rm supp}}\frac{\|d(\widehat{P}_n(\cdot|x),p(\cdot|x))\|_{\LL,\infty}}{\|d(\tilde{p}(\cdot|x),p(\cdot|x))\|_{\LL,\infty} + \left\|\left\{\sqrt{\frac{\log n}{\pi_\ell(T(x))\,n}}\right\}_{\ell=1}^{\ELL}\right\|_{\LL,\infty}}\right)\leq C.\end{eqnarray*}

\end{corollary}

This is a consequence of the previous Theorem \ref{thm:adapt} because any $\tilde{p}$ that is constant on the leaves of ${T}$ will make errors that are proportional to the continuity rates at those leaves. Therefore adapting its precision to different parts of the tree. 
Alternatively, we could compute a different estimators for the context tree for each process, namely $\widehat T_{n,\ell}$ for $\ell=1,\ldots,\LL$. Under the stated conditions both approaches lead to the same rate of convergence and the pruning rules imply that $\widehat T_n \subset \cup_{\ell=1}^\LL \widehat T_{n,\ell}$. The potential advantage of the group approach is to provide a single context tree that is applicable to all processes. However, under different choices of $\k$ exploiting the group context tree can lead to improvements as discussed earlier (see Examples in Section \ref{Sec:Examples}).

The proof of Theorem \ref{thm:adapt} shows that (\ref{Condition:n}) suffices as a requirement for the empirical frequency of any leaf $w \in T$ in the sample to be close to its expected frequency, for any given $T\in \sT(h,\pi_*)$. In the next section we consider important classes of processes that fall within this $\beta$-mixing framework.

\section{Rates of convergence for theoretical examples}\label{Sec:Examples}

In what follows we apply the finite sample analysis from the previous section to obtain asymptotic results for some classes of processes. Throughout this section we assume that $\sS$ and $A$ are fixed. The sample size $n$ diverges, and for each $n$ we have parameters $\delta^{(n)},\delta_0^{(n)},\ELL^{(n)}$ and processes
$$X^{(n)}(1),\dots, X^{(n)}(\ELL^{(n)}).$$ We impose the restrictions $$\delta^{(n)}+\delta^{(n)}_0=\bigoh{n^{-\xi}}\mbox{ and }\ELL^{(n)}\,(\delta^{(n)}\delta^{(n)}_0)^{-1}=\bigoh{n^{\alpha}}$$ for constants $\alpha\geq \xi>0$. For each example we make mixing assumptions that we assume to hold uniformly in $n$ and $\ell$. We will omit the superscript ${(n)}$ from our notation.

\subsection{Parametric case}In our first example we assume that the true model for the $\ELL$ processes has a finite context tree $T^*$, which is allowed to vary with $n$. For a fixed $n$, this implies the $\ELL$ processes are finite Markov chains, thus exponentially $\phi$-mixing; we assume {\em uniform} exponential $\beta$-mixing over all processes and all values of $n$. We also assume that $T^*$ is a {\em complete tree}, meaning that any node has $0$ or $|A|$ children (cf. Remark \ref{rem:completetrees} in the Supplementary Material for some comments on this condition which is needed just to achieve uniqueness of the context tree).

\begin{example}[Parametric Case]\label{Ex:Parametric} The processes $X(1),\dots,X(\ELL)$ are stationary and ergodic. Moreover, there exists a finite complete tree $T^*$ and transition probabilities $p=(p_1,\dots,p_\ELL)$ that are compatible with the processes:
$$\forall 1\leq \ell\leq \ELL,\,\forall a\in A\,:\, \Pr{X(\ell)_0=a\mid X(\ell)_{-\infty}^{-1}} = p_\ell(a\mid T^*(X(\ell)^{-1}_{-\infty}))\,\mbox{a.s.}.$$
Moreover, each of these processes is stationary $\beta$-mixing with the same exponential rate function:
$$\beta(b) = \chi\,e^{-\nu\,b} $$
where $\chi,\nu>0$ are independent of the sample size. We assume that $h_{T^*}\pi_{T^*}^{-1}=\liloh{n/\log n}$ and $\pi^{-1}_{T^*}=\bigoh{n^{1-\eps}}$ for some $\eps>0$.

Finally, we define $d_{T^*}\equiv 1$ if $T ^*= \{e\}$; otherwise, when $T^*\neq \{e\}$, we set $$d_{T^*}\equiv \inf\limits_{w\text{ leaf of $T^*$}, w\neq e}\left\{\sup\limits_{w'\succeq {\rm par}(w)\text{ leaf of $T^*$}}\,d(p(\cdot|w),p(\cdot|w'))\right\}.$$ We assume
$$d_{T^*}^{-1} = \liloh{\sqrt{\frac{\pi_{T^*}n}{\log n}}}.$$\end{example}

Note that, by Remark \ref{rem:completetrees} in the Supplementary Material, $d_{T^*}>0$ is equivalent to requiring that $T^*$ is the unique minimal complete context tree compatible with the processes $X(1),\dots,X(\LL)$. Our analysis implies that the ``leaf separation quantity" $d_{T^*}$ above is an appropriate detection threshold.

We have the following result.
\begin{theorem}\label{thm:parametric} In the parametric case considered in Example \ref{Ex:Parametric}, with probability $1-\bigoh{n^{-\xi}}$ we have $\widehat{T}_n=T^*$ and
$$\sup_{x\in {\rm supp}}\|d(\widehat{P}_n(\cdot|x),p(\cdot|x))\|_{\LL,\k} = \bigoh{\sqrt{\frac{\log n}{\pi_{T^*}n}}}.$$
Moreover, the $\log n$ term in the error estimate may be improved to $\log \log n$ in the ``many processes"~case of Theorem \ref{thm:main}.\end{theorem}

Remark \ref{rem:BW} in the Supplementary Material \cite{SM-AGCT} shows that this compares favourably with the theorem of B\"{u}hlmann and Wyner \cite{BuhlmannWyner1999} for the case $\ELL=1$.

\subsection{Chains with infinite connections}\label{sec:complete} In our second example we allow for infinite order chains, but require a non-nullness condition and polynomial uniform continuity.

\begin{example}[Chains with infinite connections]\label{Ex:Infinite} The processes $X(1)$, $\dots$, $X(\ELL)$ are stationary and ergodic. There exist constants $\eta>0$,$\bgamma>1+2\alpha$ and $\Gamma_0>0$ not depending on the sample size $n$ such that for all $1\leq \ell\leq \ELL$,
\[\text{\bf (non-nullness)}\,:\, \inf_{a\in A, x\in A^{-1}_{-\infty}} p_\ell(a|x)\geq \eta\]and
$$\forall k\in\N\, \max_{w\in A^{-1}_{-k}}\|\gamma(w)\|_{\LL,\k}\leq \Gamma_0\,k^{-1-\bgamma}.$$\end{example}

In this case we have the following uniform bound.
\begin{theorem}\label{thm:complete} In the case of chains with infinite connections considered in Example \ref{Ex:Infinite}, we have  $$\Pr{\sup_{x\in {\rm supp}}\|d(\widehat{P}_n(\cdot|x),p(\cdot|x))\|_{\LL,\k}=\bigoh{\frac{1}{\log^{1+\bgamma}n}}} = 1-\bigoh{n^{-\xi}}.$$\end{theorem}
This result shows that $\widehat{P}_n(\cdot|x)$ converges to $p(\cdot|x)$ uniformly over pasts $x$, albeit at a slow rate $1/\log^{1+\theta}n$. Section \ref{rem:minimaxcomplete} in the Supplementary Material \cite{SM-AGCT} shows that this is the minimax rate for uniform convergence over pasts, when $\ELL=1$ and $A=\{0,1\}$. Nonetheless, because of the adaptivity of the estimator, faster rates of convergence would be achieved for pasts with better continuity rates.

\subsection{Renewal processes} Our last example consists of stationary binary renewal processes whose arrival distributions have uniformly bounded $2+\bgamma$ moments, $\bgamma>0$.

\begin{example}[Renewal processes]\label{Ex:Renewal} Each process $X(\ell)$ is a stationary and ergodic binary renewal process. The arrival distributions  $\mu_\ell$ have support on the whole of $\N$ and satisfy $$\sum_{k\in\N}\mu_\ell(k)\,k^{2+\bgamma}\leq C$$ for constants $C,\bgamma>0$ that do not depend on $1\leq \ell\leq \ELL$ or on the sample size. Moreover,  there exist values $\{f_\ell\}_{\ell=1}^{\ELL}$ (possibly depending on $n$) such that
\begin{equation}\label{eq:renewalcont}f_\ell =\lim_{k\to +\infty} \frac{\mu_\ell(k)}{\sum_{j\geq k}\mu_\ell(j)}.\end{equation}\end{example}

In this example we have no control over the continuity rates of the process at arbitrarily deep levels. We establish the following result.
\begin{theorem}\label{thm:renewal}
In the case of renewal processes as in Example \ref{Ex:Renewal}, let $G\subset A^{-1}_{-\infty}$ be the subset of all strings $x=\dots\,10^{s-1}$ where $s$ is such that
$$\min_{1\leq \ell\leq \LL}\sum_{j\geq s}\mu_\ell(j)\geq n^{-\frac{\bgamma}{\bgamma+1}}\log n.$$ Then, the AGCT estimator satisfies the following with probability $1-\bigoh{n^{-\xi}}$:
$$\forall x=\dots 10^{s-1}\in G\,:\,\|d(\widehat{P}_n(\cdot|x),p(\cdot|x)\|_{\LL,\infty} \leq C\,\left\|\left\{\sqrt{\frac{\log n}{n\sum_{j\geq s}\mu_\ell(j)}}\right\}_{\ell=1}^{\LL}\right\|_{\LL,\infty}.$$\end{theorem}

Theorem \ref{thm:renewal} highlights the adaptivity of the rates of convergence. Indeed for pasts $x\in G$ that are more frequent, corresponding to larger values of $\sum_{j\geq s}\mu_\ell(j)$, a faster rate of convergence is obtained.

\section{Example of Applications to Functionals}\label{Sec:Applications}

In this section we develop two applications of the AGCT model and estimation algorithms. In both cases the main objects of interest are neither the context trees, nor the transition probabilities, but rather functionals of these quantities. In what follows we estimate these functionals based on $\widehat T_n$ and $\widehat P_n$ accounting for the estimation error and possible misspecification.
These two applications rely on different metrics and penalty functions, providing a motivation for the generality of the previous analysis.

\subsection{Discrete stochastic dynamic programming}\label{Section:SDP}

Discrete stochastic dynamic programming (DSDP) focuses on solving
structured optimization problems in which a control $u$ is chosen
from a set of discrete options $\mathcal{U}$ at time $t$ and
yields some instantaneous payoff $f(a,u)$, where $a \in A$ is the
current state. The system evolves to a state $x_{t+1}$
at period $t+1$ according to an $A$-valued random function $s(x^t_{-\infty},u)$ satisfying:
$$\Pr{s(x^t_{-\infty},u)=a} = p_u(a\mid x^{t}_{-\infty})\,\,(a\in A,\, u\in \mathcal{U}).$$
That is, the transition probabilities of $s(x^t_{-\infty},u)$ depend on the chosen control $u\in \mathcal{U}$ and (potentially) the complete history of states $x^t_{-\infty}\in A^{-1}_{-\infty}$.

In applications, the main object of interest is the value function that characterize the expected future payoffs as a function of the history of states:
$$ V(x) = \max_{u \in \mathcal{U}} \{ f(x_{-1},u) + \disc \Ex{V(x\,s(x,u))} \} $$
where $\disc<1$ is the discount factor and $x\,s(x,u)$ is the concatenation of $x$ with $s(x,u)$.
In practice the transition probabilities between states need to be estimated.  However, even if transition probabilities were known a priori, the tractability of a dynamic programming formulation relies on avoiding a large state space (in this case potentially  $A^{-1}_{-\infty}$). Nonetheless the selected state space needs to be rich enough to capture the main features of the transition function $s(\cdot,\cdot)$.

Our motivation to apply the AGCT estimator is to create estimates for the transition probabilities while maintaining a data-driven manageable state space. This is exactly the case in which using the AGCT model can be more attractive than using a (potentially much larger) compatible tree $T^*$. We advocate in favor of a small approximation error (that is comparable with the noise in the estimation) with a substantially smaller state space. Thus, for $x \in A^{-1}_{-\infty}$, we propose to estimate the value function with
$$ \widehat V(x) = \widehat V(\widehat T_n(x)), $$
and the transition probabilities with $\hat p_{n,u}(\cdot\mid \widehat T_n(x))=\widehat P_{n,u}(\cdot\mid x)$, which are allowed to depend on the action $u \in \mathcal{U}$.
The total number of states of the estimated system is the number of leaves of $\widehat T_n$.




Let the number of groups $\LL = |\mathcal{U}|$, $d_\ell = \|\cdot\|_1/2$ and $\k = r = \infty$. The data consists of $|\mathcal{U}|$ time series of length $n$ where on each series the decision is chosen to be constant $u\in \mathcal{U}$.
\begin{theorem}[Value Function Approximation]\label{thm:DP}
In the discrete stochastic dynamic programming problem, by choosing $\conf$ as in Definition \ref{def:firstchoice}, we have that with probability at least $1-\delta$ the estimator $\widehat V$ of the value function satisfies
{\small $$ \sup_{x \in {\rm supp}} \frac{| \widehat V(x) - V(x) |}{\sup_{a\in A}|V(xa)|{\displaystyle \inf_{T}}\left\{  \|\gamma(T(x))\|_{\LL,1}+\RR{\infty}{\conf(T(x))}\right\}} \leq \frac{\disc}{1-\disc} 4c\frac{c+1}{c-1} $$}
where
$ \displaystyle \RR{\infty}{\conf(T(x))} \lesssim \max_{\ell=1,\ldots,\LL}\sqrt{\frac{\log(n\LL/\delta) + |A|}{N_{n-1,\ell}(T(x))}}.$
\end{theorem}

As before the estimator enjoys adaptivity. In particular if we restrict the minimum above to typical trees we have the following corollary.

\begin{corollary}[Value Function Approximation for $\beta$-mixing] Under the same assumptions of Theorems \ref{thm:adapt} and \ref{thm:DP} with probability at least $1-\delta-\delta_0$
{\small $$
\sup_{T\in\sT(h,\pi_*)} \sup_{x \in A^{-1}_{-\infty}} \frac{| \widehat V(x) - V(x) |}{\sup_{a\in A}|V(xa)| \left\{ \|\gamma(T(x))\|_{\LL,1}+{\displaystyle \max_{\ell=1,\ldots,\LL}}\sqrt{\frac{\log(n/\delta)}{n\pi_\ell(T(x))}} \right\} } \leq \frac{C\disc}{1-\disc}$$}
\end{corollary}


\subsection{Dynamic discrete choice models}

In dynamic discrete choice models a group of agents makes choices among the same set of options over time \cite{ArellanoHonore2001,ChernozhukovFernandez-ValHahnNewey2009,BrowningCarro2010,AguirregabiriaMira2010,BrowningCarro2011}. Models usually pre-specify a Markovian structure of the process, which is commonly assumed to be of order 1. We are interested in relaxing this assumption and to estimate the relevant context tree and the associated transition probabilities.

Agents are assumed to be sampled independently from the same population. We assume that the underlying context tree is the same across agents, but allow for the specific transition probability to vary by agent to account for heterogeneity. Herein we focus on the case with no covariates, but results can be extended to the case of discrete covariates \cite{BrowningCarro2010,BrowningCarro2011}.

In applications, the main interest is on statistics that are functions of the  conditional probabilities rather than the conditional probabilities themselves.
Here we focus on the average marginal dynamic effect for $a\in A$, $x,y\in A^{-1}_{-\infty}$ $${\rm AVEm}(a,x,y) = \Ex{m_\ell(a,x,y)}$$ where the marginal dynamic effect
$m_\ell(a,x,y) = p_\ell( a | x ) - p_\ell(a|y),$ and the expectation is taken over the distribution of agents in the population of interest.
The average marginal dynamic effect measures the average over the population of the change in the probability of selection of an option $a\in A$ between two different histories of past consumption $x, y \in A^{-1}_{-\infty}$.
Other measures of interest in the literature are the long run proportions of a particular option being chosen, or the probability of selecting a particular option $t$ periods ahead given the current state, see  \cite{BrowningCarro2010}. 


The estimator of the marginal dynamic effect for an option $a\in A$ and histories of consumptions $x,y\in A^{-1}_{-\infty}$ for the $\ell$th agent is $$ \hat m_\ell(a,x,y) =  \hat p_{n,\ell}( a | \widehat T_n(x) ) - \hat p_{n,\ell}(a|\widehat T_n(y)), $$ and the estimator for the average marginal dynamic effect is
$$ \widehat{{\rm AVEm}}(a,x,y) = \frac{1}{\LL}\sum_{\ell=1}^\LL \hat m_\ell(a,x,y).$$
Therefore, if the conditional probabilities were known, a rate of $1/\sqrt{L}$ would be optimal for the estimation of a single average marginal dynamic effect. In what follows we will use the AGCT model to estimate these dynamic effects uniformly over all histories. This motivates the choice of $d_\ell = \|\cdot\|_\infty$, $\k=1$, and $r=2$ in the AGCT estimator.


\begin{theorem}\label{Thm:DCM}
In the dynamic discrete choice model, if the context tree and conditional probabilities are estimated with $\conf$ as in Definition \ref{def:secondchoice}, we have that with probability at least $1-2\delta$ the estimator for the average marginal dynamic effect satisfies
{\small $$\sup_{ { a\in A, \atop x,y \in {\rm supp}}} \frac{| \widehat{{\rm AVEm}}(a,x,y) - {\rm AVEm}(a,x,y) |}{ {\displaystyle \max_{z=x,y}\inf_T} \left\{ \|\gamma(T(z))\|_{\LL,1}+\RR{2}{\conf(T(z))} \right\} + \sqrt{\frac{2\log\left(\frac{|A|\cdot n^4}{4\delta}\right)}{\LL}}+\frac{2}{\LL}} \leq  8c\frac{c+1}{c-1}.$$}
where $\RR{2}{\conf(T(z))}\lesssim \sqrt{\frac{\log \log n + \log |A|}{\LL}\sum_{\ell=1}^\LL  1/ N_{n-1,\ell}(T(z))}$, $z\in A^{-1}_{-\infty}$.
\end{theorem}

This uniform rate of convergence for the average marginal dynamic effect is governed by the rate of convergence of the conditional probabilities of the best context tree estimator, and the number of different agents in the data. 
Interestingly, the above result holds uniformly over all pairs $x, y\in A^{-1}_{-\infty}$. 


%
%

\section{Linguistic rhythm differences between European and Brazilian Portuguese}\label{Sec:Linguistic}

In this section we revisit the application and the data considered in \cite{GGGL2012} regarding the linguistic features underlying the European Portuguese (EP) and Brazilian Portuguese (BP) languages. The goal of \cite{GGGL2012} was to compare the rhythmic fingerprints of the two languages in written form.

For each language, the data consist of articles from a popular daily newspaper from the years 1994 and 1995. For each year and each newspaper, 20 articles were randomly selected. The linguistic features are represented by a quinary alphabet with four rhythmic features (0,1,2,3) and an additional feature representing the end of an article (4). The four rhythmic features represent: non-stressed, non prosodic word initial syllable (0); stressed, non prosodic word initial syllable (1); non-stressed, prosodic word initial syllable (2); and stressed prosodic word initial syllable (3) Each data sample was then treated as a stochastic process, and a variant of the BIC model selection method was used to fit a context tree to each sample. Their main finding was summarized as follows.

\begin{quotation}[T]he main difference between the two languages is that whereas in BP both 2 (unstressed boundary of a phonological word) and 3 (stressed boundary of a phonological word) are contexts, in EP only 3 is a context. This means that in EP, as far as noninitial stress words are concerned, the choice of lexical items is dependent on the rhythmic properties of the preceding words. This is not true when the word begins with a stressed syllable. This does not occur in BP, where word boundaries are always contexts, and as such insensitive to what occurs before, independently of being stressed or not. These statistical findings are compatible with the current discussion in the linguistic literature concerning the different behavior of phonological words in the two languages [...] (Galves et. al., \cite[Section 6]{GGGL2012})\end{quotation}

In \cite{GGGL2012}, for each newspaper, the 40 days sample is concatenated into a single string containing respectively a sequence of 105326 and 97750 linguistic features. In order to concatenate articles from different days, a homogeneity assumption was required. However, heterogeneity over different days, or at least over the different years are a source of potential concern. For example, 1994 was a World Cup year and the media in both countries are heavily influenced by such event. Our own study accounts for possible heterogeneity on the conditional probabilities by treating each year as a group in the group context tree model. Thus we allow for year specific conditional probabilities.

Figure \ref{Fig:2groups} displays the estimated context trees. Our findings are in good agreement with \cite{GGGL2012}, in that the context trees found for BP in both studies are the same, and our tree for EP strictly contains the one found in \cite{GGGL2012}. In particular, we corroborate their finding that 2 is a context for BP but {\em not} for EP.

\begin{figure}[h!]
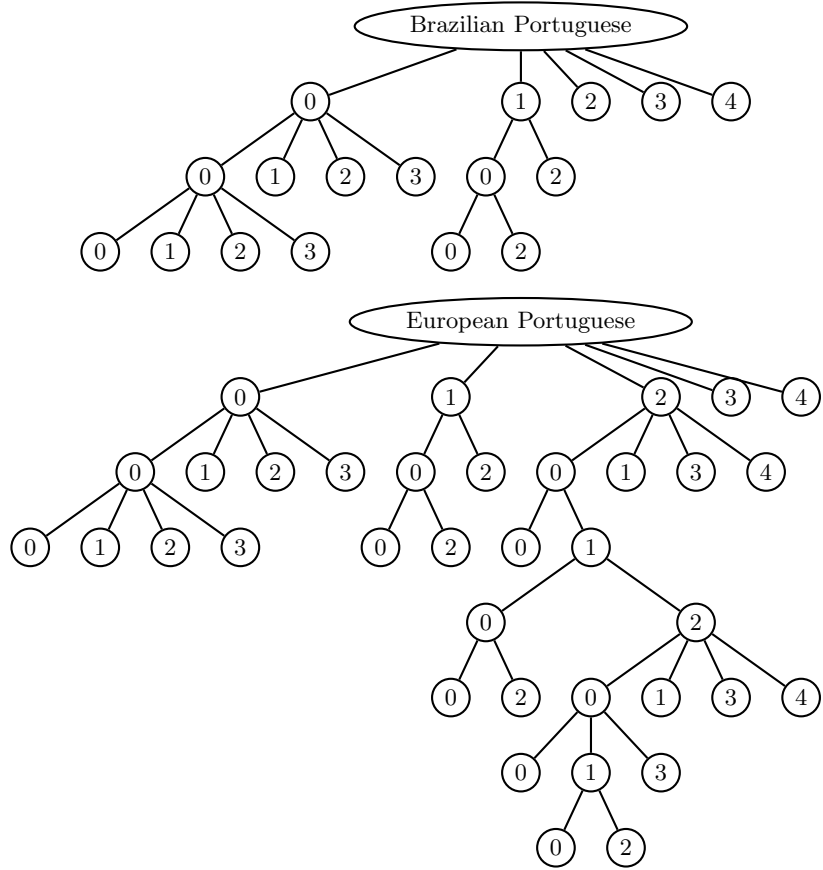

\begin{center}
\pstree[levelsep=1cm,treesep=0.4cm]{\Toval{Brazilian Portuguese}}
{
\pstree{\Tcircle{0}}
        {
        \pstree{\Tcircle{0}}
        {\Tcircle{0} \Tcircle{1} \Tcircle{2} \Tcircle{3}}
        \Tcircle{1}
        \Tcircle{2}
        \Tcircle{3}
        }
\pstree{\Tcircle{1}}
        {
        \pstree{\Tcircle{0}}
            {\Tcircle{0} \Tcircle{2}}
        \Tcircle{2}
        }
\Tcircle{2}
\Tcircle{3}
\Tcircle{4}
}

\vspace{0.3cm}
\pstree[levelsep=1cm,treesep=0.4cm]{\Toval{European Portuguese}}
{
\pstree{\Tcircle{0}}
        {
        \pstree{\Tcircle{0}}
        {\Tcircle{0} \Tcircle{1} \Tcircle{2} \Tcircle{3}}
        \Tcircle{1}
        \Tcircle{2}
        \Tcircle{3}
        }
\pstree{\Tcircle{1}}
        {
        \pstree{\Tcircle{0}}
            {\Tcircle{0} \Tcircle{2}}
        \Tcircle{2}
        }
\pstree{\Tcircle{2}}
        {
        \pstree{\Tcircle{0}}
            {
            \Tcircle{0}
            \pstree{\Tcircle{1}}
                {
                \pstree{\Tcircle{0}}
                    { \Tcircle{0}  \Tcircle{2} }
                \pstree{\Tcircle{2}}
                    {
                    \pstree{\Tcircle{0}}
                        {
                        \Tcircle{0}
                        \pstree{\Tcircle{1}}
                            { \Tcircle{0} \Tcircle{2} }
                        \Tcircle{3}
                        }
                    \Tcircle{1}\Tcircle{3}\Tcircle{4}
                    }
                }
            }
        \Tcircle{1}\Tcircle{3}\Tcircle{4}
        }
\Tcircle{3}\Tcircle{4}
}
\end{center}\caption{Estimated context trees for the Brazilian Portuguese and European Portuguese languages based accounting for heterogeneity in different years.}\label{Fig:2groups}
\end{figure}

\section{Discussions and Variations}\label{Sec:Conclusion}

\subsection{Comparisons with Single-process case}

We briefly indicate similarities and differences between the results presented above with \cite{FerrariWyner2003}, which concerns the single-process case. The work \cite{FerrariWyner2003} proves weak consistency in the estimation of conditional probabilities and of (truncated) context trees for all nodes in a tree $T_n$ that grows with the sample size $n$. For this they assume that the stochastic process is geometrically $\alpha$-mixing, and also that there is sufficient separation between the conditional probabilities corresponding to leaves of the tree and their parents. The authors of \cite{FerrariWyner2003} point out that the latter assumptions might be hard to check in practice.

Our analysis differs from theirs in several important aspects even in the case of $\ELL=1$ processes. Our goal is to estimate transition probabilities given the entire infinite past, uniformly over all such pasts. Achieving consistency in our setting requires that these probabilities be continuous functions of the infinite past, which \cite{FerrariWyner2003} do not need to assume. By contrast, given continuity and $\beta$-mixing, model selection and probability estimation become separate tasks. In particular, our results on the transition probabilities do not require any kind of separation between leaves and their parents. In addition, our results cover natural and interesting classes of processes (such as certain renewal processes) where geometric mixing bounds are not available. Other points of the analysis are mostly incomparable due to the differences in assumptions.

\subsection{Computational Efficiency and Variations}

The algorithm can be implemented efficiently, i.e. in polynomial time with respect to the parameters $\LL$ and $n$ of the data. Observe that ${\sf CanRmv}(w)$ can be computed efficiently from the list of values:
$${\rm List}(w)\equiv \{(\hat{p}_n(\cdot|w'),\conf(w'))\,:\, w'\in E_n,\, w'\succeq w\}$$
and the corresponding list for $\parent(w)$. Since ${\sf CanRmv}(w)$ is only computed for leaves of the current tree $\widehat{T}_n$,  we only need to ensure that at all times, each leaf node and each parent of a leaf stores the correct list ${\rm List}(w)$. This can be achieved as follows.
\begin{itemize}
\item initially, one sets ${\rm List}(w)=\{\hat{p}_n(\cdot|w),\conf(w)\}$ for each $w\in E_n$;
\item whenever a leaf $w$ is examined in $\widehat{T}_n$, its parent's list is updated:
$${\rm List}({\rm par}(w))\leftarrow  {\rm List}({\rm par}(w))\cup \left(\bigcup\limits_{w'\in \widehat{T}_n\,:\,{\rm par}(w')={\rm par}(w)} {\rm List}(w')\right).$$
Actually, this update only needs to be performed at the first time a child of ${\rm par}(w)$ is examined.\end{itemize}We note in passing that more efficient algorithms can be found for the case $\ELL=1$ with the $\ell_\infty$ metric by using compact suffix trees. This will be elaborated upon in a companion paper.

All results established in this work would remain valid if in the definition of  ${\sf CanRmv}(w)$ in (\ref{Def:IsPrunable})
we set $w'' \in \mathcal{W}$ where $\parent(w) \in \mathcal{W} \subseteq \{ z \in E_n : z \succeq \parent(w)\}$. For the same choice of confidence radius, computationally we would like to use the smallest set $\mathcal{W}$ while statistically we would like to use the largest such set.

\subsection{Improvement on confidence radii based on maximal variance}

The choices of confidence radii described in Definition \ref{def:firstchoice} and \ref{def:secondchoice} do not explore the intrinsic variance within the norm, namely
$$ \bar{\sigma}_\ell^2(w) := \max_{S\in\sS} \ \bar{p}_{n,\ell}(S\mid w)(1-\bar{p}_{n,\ell}(S\mid w))$$
where $\bar{p}_{n,\ell}(S \mid w)$ is a weighted sum of probabilities defined in (\ref{eq:deforacle}) for which $\hat p_{n,\ell}(S\mid w)$ is a consistent estimator. (These probabilities can be seen as an oracle estimator, see Section \ref{Sec:Oracle} in the Supplementary Material \cite{SM-AGCT} for a discussion.) Generically, adding variance to our bounds does not necessarily improve rates of convergence but can improve finite sample performance, particularly in the case of $d_\ell = \|\cdot\|_\infty$ with $|A|>2$. Here we discuss such a modification of Definition \ref{def:firstchoice}  that leads to strictly smaller confidence radii while still achieving the same guarantees as in Theorem \ref{thm:main}. However, the variance-based control can be applied to a suffix $w$ only if there were enough occurrences of the suffix in the data, namely the following event occurred
$$J_{\ell,w} :=\left\{ N_{n-1,\ell}(w)  \geq \frac{2\log(n^2|\sS|/\delta) + 4\log[2+2\log(\bar \sigma_\ell^2(w)N_{n-1,\ell}(w))]}{\bar \sigma_\ell^2(w)\log^2(3/2)}\right\}.$$
Otherwise, we use the previous choice as in Definition \ref{def:firstchoice}. To concisely state the results regarding the maximum variance we define
$$\widetilde \sigma_\ell(w) :=  \sqrt{2}\bar \sigma_\ell(w)\Ind{\{ J_{\ell,w}\}} + \Ind{\{ J_{\ell,w}^c\}} \leq 1.$$
We define $\conf^{\tilde{\sigma}}_\ell(w):= \widetilde \sigma_\ell(w)\conf_\ell(w)$. By construction, it follows that $\conf^{\tilde{\sigma}}_\ell(w)\leq \conf_{\ell,m}(w)$ since $\widetilde \sigma_\ell(w) \leq 1$. However, $\conf^{\tilde{\sigma}}_\ell(w)$ might not be non-increasing in $w$. Nonetheless, the confidence radius $\conf^{\tilde{\sigma}}_\ell(w)$ can be majorated by the monotone confidence radius which still leads to an improvement over $\conf_{\ell,m}(w)$, namely
$$ \conf^*_\ell(w) = \max_{w'\preceq w} \conf^{\tilde{\sigma}}_\ell(w') \leq \max_{w'\preceq w} \conf_{\ell,m}(w') = \conf_{\ell,m}(w).$$
A side remark is that  $\conf^*_\ell(w)$ requires the estimation of $\bar\sigma_\ell(w)$. Indeed, the estimates need to satisfy $\bar\sigma_\ell(w) \leq \hat\sigma_\ell(w)$ with high probability uniformly over $w\in E_n$. However, it follows that any such estimator will satisfy $\hat\sigma_\ell(w)\leq 1/2$ so that even by setting $\hat\sigma_\ell(w)= 1/2$ we still achieve smaller confidence radius than the original definition.

\section{Conclusion}

Understanding the memory structure of stochastic processes has proved to be of fundamental importance in applications. VLMC models have been playing a central role in modeling and estimating  stationary processes with discrete alphabets. In this work we consider an extension of the traditional VLMC in which many stationary processes share the same context tree but potentially different conditional probabilities. Since we allow for potentially infinite memory processes, we propose to focus the estimation on an oracle context tree that optimally balances the bias and variance trade-off for a given sample.

We propose a computationally efficient estimator for the underlying context tree and the associated conditional probabilities. We establish several properties of the proposed estimator, including adaptivity and oracle inequalities for the estimation of conditional probabilities. We propose and analyze data-driven choices of the penalty parameters for the regularization, and study its typical behavior under $\beta$-mixing conditions. Two applications, discrete dynamic stochastic programming and discrete choice models, motivated the proposal of the AGCT model. In these applications we are interest in functionals of the conditional probabilities.   We developed the uniform bounds for the estimation of these functionals accounting for possible misspecification of the estimated context tree.

Finally, we investigate the application of the group context tree model and the proposed estimators to investigate the rhythmic differences between Brazilian and European Portuguese allowing for possible heterogeneity in the sample. Our results fully support previous findings of the literature.

\appendix


\section{Proof of Theorem 1}\label{App:ProofsSec4}

Theorem \ref{thm:main} follows directly from three Lemmas related to the {\em good event} $\Good_*$:
\begin{equation}\label{eq:defgood1}\Good_*\equiv \bigcap_{w\in A^*}\left\{\begin{array}{l}\|d(p(\cdot|w),\widehat{p}_n(\cdot|w))\|_{\ELL,\k}\\ \leq \|\{\gamma_\ell(w)\}_{\ell=1}^{\ELL}\|_{\ELL,\k} + \|\conf(w)\|_{\ELL,r}\end{array}\right\},\end{equation}
where $r=\k\in [1,+\infty]$ in the ``general case" and $r=2$, $\k=1$ in the ``many processes"~case of Theorem \ref{thm:main}, and for $|w|<\infty$ we define $p_\ell(a|w)= \Pr{X_0(\ell)=a\mid X_{-|w|}^{-1}(\ell)=w}$, if $\Pr{X_{-|w|}^{-1}(\ell)=w}>0$, and $p_\ell(a|w)=1/|A|$ if $\Pr{X_{-|w|}^{-1}(\ell)=w}=0$.

\begin{lemma}[Proven in Section \ref{Ass:Sec4Contain}]\label{lem:contain}If $\Good_*$ holds, $\widehat{T}_n\subset T^*$.\end{lemma}
\begin{lemma}[Proven in Section \ref{Ass:Sec4Adapt}]\label{lem:adapt} If $\Good_*$ holds, then for all $x\in A^{-1}_{-\infty}$ and any finite tree $T$
$$\|d(p(\cdot|x),\widehat{P}_n(\cdot|x))\|_{\ELL,\k}\leq  \frac{2c+2}{c-1}\,\|\{\gamma_\ell(T(x))\}_{\ell=1}^{\ELL}\|_{\ELL,\k} + (1+2c)\|\conf(T(x))\|_{\ELL,r}.$$\end{lemma}
\begin{lemma}[Proven in Section \ref{Ass:Sec4AProbGood}]\label{lem:probgood}The probability of $\Good_*$ is $\geq 1-\delta$.\end{lemma}

These three lemmas are proven subsequently.

\subsection{Proof of Lemma \ref{lem:contain}}\label{Ass:Sec4Contain}

Let $z\in A^*\backslash T^*$ and assume $\Good_*$; we will show that $z\not\in\widehat T_n$. Let $w$ be an ancestor of $z$ which is a leaf of $T^*$. Because $T^*$ is the true context tree, $ \|\{\gamma_\ell(w')\}_{\ell=1}^{\ELL}\|_{\ELL,\k}=0$ for all descendants of $w$, in particular for $z,\parent(z)$ and their descendants. If we assume $\Good_*$ holds, the triangle inequality gives
$$\forall u,v\succeq \parent(z)\,:\,\|d(\widehat{p}_n(\cdot|u),\widehat{p}_n(\cdot|v))\|_{\ELL,\k}\leq  \|\conf(u)\|_{\ELL,r} +   \|\conf(v)\|_{\ELL,r},$$
and one can easily deduce from this that ${\rm CanRmv}(u)=1$ for all $u\succeq z$. This means $z$ is pruned from the tree.

\subsection{Proof of Lemma \ref{lem:adapt}}\label{Ass:Sec4Adapt}

Fix $x$ and $T$. Recall that $\widehat{P}_n(\cdot|x)=\widehat{p}_n(\cdot|\widehat T_n(x))$ and that $\|\conf(w)\|_{\ELL,r}$ is monotone non-decreasing in $w$. Notice that $\widehat T_n(x)$ and $T(x)$ are both finite suffixes of $x$. This allows us to divide the analysis into three cases.

{\bf Case 0: $\widehat T_n(x)=T(x)$.} The result follows from
$$\begin{array}{rl}
& \|d(p(\cdot|x),\widehat{p}_n(\cdot|\widehat T_n(x)))\|_{\ELL,\k}  =  \|d(p(\cdot|x),\widehat{p}_n(\cdot|T(x)))\|_{\ELL,\k} \\
& \leq \|d(p(\cdot|x),p(\cdot|T(x)))\|_{\ELL,\k}+\|d(p(\cdot|T(x)),\widehat{p}_n(\cdot|\widehat T_n(x)))\|_{\ELL,\k}\\
& \leq 2\|\{\gamma_\ell(T(x))\}_{\ell=1}^{\ELL}\|_{\ELL,\k} + \|\conf(T(x))\|_{\ELL,r}\end{array}$$
where the first equality is from $\widehat T_n(x)=T(x)$, the second step from triangle inequality, and the third from the event $\Good_*$ and the definition of the continuity rates.\\

{\bf Case 1: $\widehat T_n(x)\prec T(x)$.} Let $w$ denote the child of $\widehat T_n(x)$ on the path to $T(x)$. Note that $w$ must have been pruned, otherwise $w\in \widehat{T}_n$ would be a longer suffix of $x$ than $\widehat{T}_n(x)$.

We deduce that $w$ satisfies ${\sf CanRmv}(w)=1$, like any other pruned node. In particular, this implies that $T(x)\succeq w$ and $\widehat T_n(x)=\parent(w)$ satisfy
$$\|d(\widehat{p}_n(\cdot|T(x)),\widehat{p}_n(\cdot|\widehat T_n(x)))\|_{\ELL,\k}\leq c\,[\|\conf(T(x))\|_{\ELL,r} +   \|\conf(\widehat T_n(x))\|_{\ELL,r}].$$
Since $\widehat T_n(x)\prec T(x)$ the RHS of the above display is $\leq 2c\,\|\conf(T(x))\|_{\ELL,r}$, and the occurrence of $\Good_*$ gives
$$\|d(p(\cdot|T(x)),\widehat{p}_n(\cdot|T(x)))\|_{\ELL,\k}\leq \|\{\gamma_\ell(T(x))\}_{\ell=1}^{\ELL}\|_{\ELL,\k}+\|\conf(T(x))\|_{\ELL,r}.$$
Combining these observations and employing the triangle inequality gives:
$$\|d(p(\cdot|T(x)),\widehat{p}_n(\cdot|\widehat T_n(x)))\|_{\ELL,\k}\leq \|\{\gamma_\ell(T(x))\}_{\ell=1}^{\ELL}\|_{\ELL,\k}+(1+2c)\,\|\conf(T(x))\|_{\ELL,r}.$$
Using $\|d(p(\cdot|x),p(\cdot|T(x)))\|_{\ELL,\k}\leq \|\{\gamma_\ell(T(x))\}_{\ell=1}^{\ELL}\|_{\ELL,\k}$ and another application of the triangle inequality finishes the proof in this case.\\

{\bf Case 2: $\widehat T_n(x)\succ  T(x)$.} We make the following claim.
\begin{claim}[Proven subsequently]\label{Claim:1}$$\|\conf(\widehat T_n(x))\|_{\ELL,r} +   \|\conf(T(x))\|_{\ELL,r}\leq \frac{3}{c-1}\,\|\{\gamma_\ell(T(x))\}_{\ell=1}^{\ELL}\|_{\ELL,\k}.$$\end{claim}
To see how the claim implies the result, we note that
\begin{eqnarray*}\|d(p(\cdot|x),\widehat{p}_n(\cdot|\widehat T_n(x)))\|_{\ELL,\k}&\leq&  \|d(p(\cdot|x),p(\cdot|\widehat T_n(x)))\|_{\ELL,\k} \\ & & +  \|d(p(\cdot|\widehat T_n(x)),\widehat{p}_n(\cdot|\widehat T_n(x)))\|_{\ELL,\k}\\ \mbox{(use continuity rates)} &\leq &  \|\{\gamma_\ell(\widehat T_n(x))\}_{\ell=1}^{\ELL}\|_{\ELL,\k}  \\ & & +  \|d(p(\cdot|\widehat T_n(x)),\widehat{p}_n(\cdot|\widehat T_n(x)))\|_{\ELL,\k}\\  \mbox{($\Good_*$ holds)} &\leq &   2\|\{\gamma_\ell(\widehat T_n(x))\}_{\ell=1}^{\ELL}\|_{\ELL,\k} \\ & &+ \|\conf(\widehat T_n(x))\|_{\ELL,r} \\ \mbox{( $T(x)\preceq \widehat T_n(x)\Rightarrow \gamma_\ell(T(x))$ larger)} &\leq & 2\,\|\{\gamma_\ell(T(x))\}_{\ell=1}^{\ELL}\|_{\ELL,\k} \\ & & + \|\conf(\widehat T_n(x))\|_{\ELL,r} \\ \mbox{(use Claim)}&\leq  & \left(2 + \frac{3}{c-1}\right)\,\,\|\{\gamma_\ell(T(x))\}_{\ell=1}^{\ELL}\|_{\ELL,\k}\end{eqnarray*}

It remains to prove the claim. Since $\widehat T_n(x)$ was not pruned, there exist $w'\succeq \widehat T_n(x)$, $w''\succeq \parent(\widehat T_n(x))\succeq T(x)$ with
\begin{equation}\label{eq:reversethisman}c\,[\|\conf(w')\|_{\ELL,r} +   \|\conf(w'')\|_{\ELL,r}]<\|d(\widehat{p}_n(\cdot|w'),\widehat{p}_n(\cdot|w''))\|_{\ELL,\k}.\end{equation}
On the other hand,
\begin{eqnarray*}\|d(\widehat{p}(\cdot|w'),\widehat{p}_n(\cdot|w''))\|_{\ELL,\k}&\leq &\|d(p(\cdot|w'),p(\cdot|w''))\|_{\ELL,\k}\\ & & + \|d(\widehat{p}(\cdot|w'),p(\cdot|w'))\|_{\ELL,\k}\\ & & + \|d(\widehat{p}(\cdot|w''),p(\cdot|w''))\|_{\ELL,\k}.\end{eqnarray*}
The first term in the RHS is $\leq\|\{\gamma_\ell(T(x))\}_{\ell=1}^{\ELL}\|_{\ELL,\k}$ since $w',w''\succeq T(x)$. The other two terms can be bounded via $\Good_*$, and we obtain:
\begin{eqnarray*}\|d(\widehat{p}(\cdot|w'),\widehat{p}_n(\cdot|w''))\|_{\ELL,\k}&\leq &3\,\|\{\gamma_\ell(T(x))\}_{\ell=1}^{\ELL}\|_{\ELL,\k}\\ & & +\|\conf(w')\|_{\ELL,r} +   \|\conf(w'')\|_{\ELL,r}.\end{eqnarray*}
Combining this with (\ref{eq:reversethisman}) gives:
$$\|\conf(w')\|_{\ELL,r} +   \|\conf(w'')\|_{\ELL,r}\leq \frac{3}{c-1}\,\|\{\gamma_\ell(T(x))\}_{\ell=1}^{\ELL}\|_{\ELL,\k}.$$
The proof of the claim finishes once we recall that $w'\succeq \widehat T_n(x)$, $w''\succeq T(x)$ and the confidence radii $\|\conf(w)\|_{\ELL,r}$ are monotone functions of $w$.

\subsection{Proof of Lemma \ref{lem:probgood}}\label{Ass:Sec4AProbGood}

We define what one might call {\em oracle transition probabilities}: given a context $w \in E_n$, $ \ a \in A, \  \ell=1,\ldots,\LL$ as
\begin{equation}\label{eq:deforacle}\overline{p}_{n,\ell}(a|w)\equiv \frac{1}{N_{n-1,\ell}(w)}\sum_{i=|w|+1}^{n}\Ind{\{X^{i-1}_{i-|w|}(\ell)=w\}}\,p_\ell(a|X^{i-1}_{-\infty}(\ell)),\end{equation}
 and as $\overline{p}_{n,\ell}(a|w)\equiv 1/|A|$ if $w \notin E_n$. A salient feature is that these random transition probabilities are always close to the actual transition probabilities in the following sense:
\begin{equation}\label{eq:oraclealwaysclose}\mbox{If $T(x)\in E_n$, }\|d(p(\cdot|x),\overline{p}_n(\cdot|T(x))\|_{\ELL,\k}\leq \|\{\gamma_\ell(T(x))\}_{\ell=1}^\ELL\|_{\ELL,\k}.\end{equation}
This follows from the fact that $\overline{p}_{n,\ell}(\cdot|T(x))$ is a convex combination of transition probabilities $p_\ell(\cdot|y)$ with $y\succeq T(x)$.

To continue we choose a parameter $m=\infty$ in the ``general case"~of Theorem \ref{thm:main}, and $m=2$ in the ``many processes" case of the same Theorem. The following regularization event will be important in our analysis:
{ \begin{equation}\label{Def:Goodm}\Good_m\equiv \bigcap_{w\in E_n}\left\{ \ \mbox{ $ \RR{m}{\left\{ \frac{d_\ell(\overline{p}_{n,\ell}(\cdot|w),\hat{p}_{n,\ell}(\cdot|w))}{\conf_{\ell}(w)}  \right\}_{\ell=1}^\LL}$}\leq 1 \right\}.\end{equation}}

\begin{claim}\label{claim:Good}$\Good_m\subset \Good_*$, where $\Good_*$ was defined in (\ref{eq:defgood1}).\end{claim}
\begin{proof}[Proof] By (\ref{eq:oraclealwaysclose}) and the triangle inequality it suffices to show that we have the inequality
$$\MM{\k}{d(\hat p_n(\cdot|w), \bar p_n(\cdot|w))}\leq \RR{r}{\conf(w)}$$
for all $w\in A^*$ whenever $\Good_m$ holds. This is trivially true when $w\not\in E_n$. When $w\in E_n$, H\"{o}lder's inequality implies:
$$\MM{\k}{d(\hat p_n(\cdot|w), \bar p_n(\cdot|w))}\leq \RR{m}{\left\{ \frac{d_\ell(\overline{p}_{n,\ell}(\cdot|w),\hat{p}_{n,\ell}(\cdot|w))}{\conf_{\ell}(w)}  \right\}_{\ell=1}^\LL}\, \,\RR{r}{\conf(w)}$$
and the first term in the RHS is $\leq 1$ in $\Good_m$.\end{proof}

The remainder of the proof consists of showing that\begin{claim}\label{claim:finalmain}$\Pr{\Good_m}\geq 1-\delta$.\end{claim}
This clearly suffices to finish the proof in both cases.

We will use a martingale framework from Appendix \ref{App:Martingale} in the Supplementary Material \cite{SM-AGCT}. The following is the special case $\gamma=2$ and $i_0=\log_2 n$ of Lemma \ref{thm:economicalmgfestimate} 
in the appendix.
\begin{lemma}\label{lem:martingalemain}Let $(M_j,\sF_j)_{m=0}^n$ be a martingale with $M_0=0$. Assume that for each $1\leq j\leq n$ we have a $\sF_{j-1}$-measurable indicator random variable $Y_{j-1}$ with $|M_{j}-M_{j-1}|\leq Y_{j-1}$ almost surely, and define $V_{n}\equiv \sum_{j=0}^{n-1}Y_j^2.$ Then
$$\forall t\geq 0\,:\, \Pr{\frac{M^2_n}{4V_n}-2\ln(2+\log_2 V_n)\geq t\mid V_n>0}\leq e^{-t}.$$\end{lemma}
Recall that the metric $d=d_1=\dots=d_\ELL$ is given by
$$d_\ell(p,q) = d_\sS(p,q) = \sup_{A\in \sS}|p(A) - q(A)|.$$
We will consider a family of martingales indexed by $w\in A^*$, $S\in \sS$ and $1\leq \ell\leq \ELL$.  A simple calculation reveals that for any $j\in\N$
$$M_{j,\ell}(wS) = N_{j-1,\ell}(w)\,\sum_{a\in S}(\widehat{p}_{j,\ell}(a|w)-\overline{p}_{j,\ell}(a|w))$$ is a martingale under the natural filtration. One may take
$$Y_{j-1} = 0\mbox{ if }j-1<|w|\mbox{ and }Y_{j-1}=\Ind\,\{X^{m-1}_{m-|w|}=w\}\mbox{ otherwise},$$
so that the corresponding $V_n=V_{n,\ell}(wS)$ equals $N_{n-1,\ell}(w)\vee 1$. We also have that
$$N_{n-1,\ell}(w)\,d_\ell(\widehat{p}_{n,\ell}(\cdot|w),\overline{p}_{n,\ell}(\cdot|w))^2 = \max_{S\in\sS}\frac{M_{n}(wS)^2}{V_n(wS)}$$
whenever $N_{n-1,\ell}(w)>0$. The following is immediate from this discussion combined with Lemma \ref{lem:martingalemain}.

\begin{lemma}\label{lem:martingalecondensed}For any $w\in A^*$ and any process $1\leq \ell\leq \ELL$ we have:
$$\forall t\geq 0\,:\, \Pr{\frac{N_{n-1,\ell}(w)\,d_\ell(\widehat{p}_{n,\ell}(\cdot|w),\overline{p}_{n,\ell}(\cdot|w))^2}{4\,[2\ln(2 + \log_2N_{n-1,\ell}(w)) + \ln|\sS| + t]}>1\mid N_{n-1,\ell}(w)>0}\leq e^{-t}.$$\end{lemma}

We use this lemma to prove Claim \ref{claim:finalmain} in the two cases.

\begin{proof}[Proof of Claim \ref{claim:finalmain} in the ``general case"] Set $t:=\ln(n^2\LL/\delta)$. For $\ell=1,2,\dots,\LL$, define
\begin{eqnarray*}A_\ell&:=&\left\{\frac{N_{n-1,\ell}(w)\,d_\ell(\widehat{p}_n(\cdot|w),\overline{p}_n(\cdot|w))^2}{4\,[2\ln(2 + \log_2N_{n-1,\ell}(w)) + \ln|\sS| + t]}>1\right\};\\ B_{\ell}&:=&\{N_{n-1,\ell}(w)>0\}.\end{eqnarray*}
Lemma \ref{lem:martingalecondensed} gives $\Pr{A_{\ell}\mid B_{\ell}}\leq \delta/n^2\LL$ for each $1\leq \ell\leq \LL$. Recalling the formula for $\conf_\ell(w)$ in Definition \ref{def:firstchoice}, we see that $d_\ell(\widehat{p}_{n,\ell}(\cdot|w),\overline{p}_{n,\ell}(\cdot|w))>\conf_\ell(w)$ if and only if $A_\ell$ holds. Therefore, \begin{equation}\label{eq:caracaeleleutudo}
\begin{array}{rl}
& \Pr{\left\|\left\{\frac{d_\ell(\widehat{p}_n(\cdot|w),\overline{p}_n(\cdot|w))}{\conf_\ell(w)}\right\}_{\ell=1}^{\ELL}\right\|_{\ELL,\infty}>1\mid \min_\ell N_{n-1,\ell}(w)>0} \\
&= \Pr{\cup_{\ell=1}^{\LL} A_{\ell}\mid \cap_{\ell'=1}^{\LL}B_{\ell'}}\\
&\leq  \sum_{\ell=1}^{\LL}\Pr{A_\ell\mid \cap_{\ell'=1}^{\LL}B_{\ell'}}.\end{array}\end{equation}
Now recall that the \LL~processes $X(\ell)$ are all independent, therefore $A_\ell$ depends on $B_\ell$ but not on $B_{\ell'}$ for $\ell'\neq \ell$. We obtain\[\Pr{A_\ell\mid \cap_{\ell=1}^{\LL}B_{\ell'}} = \Pr{A_\ell\mid B_{\ell}}\leq \frac{\delta}{n^2\LL}.\] Plugging this back into (\ref{eq:caracaeleleutudo}) and removing the conditioning gives
\[\Pr{\left\|\left\{\frac{d_\ell(\widehat{p}_n(\cdot|w),\overline{p}_n(\cdot|w))}{\conf_\ell(w)}\right\}_{\ell=1}^{\ELL}\right\|_{\ELL,\infty}>1}\leq \frac{\delta}{n^2}\,\Pr{\min_{1\leq \ell\leq \LL}N_{n-1,\ell}(w)>0}.\]Taking a union bound over all $w\in A^*$ and bounding \[\Pr{\min_{1\leq \ell\leq \LL}N_{n-1,\ell}(w)>0}\leq \Pr{N_{n-1,1}(w)>0}\] gives
$$1-\Pr{\Good_\infty} \leq \frac{\delta}{n^2}\,\sum_{w\in A^*}\Pr{N_{n-1,\ell}(w)>0}.$$
The sum of probabilities in the RHS is the expected number of distinct substrings of $X_1^{n-1}(1)$, which is at most $n^2$. This implies $\Good_\infty$ occurs with probability at least $1-\delta$, as desired.\end{proof}

\begin{proof}[Proof of Claim \ref{claim:finalmain} in the ``Many processes"~case] For each $w\in A^*$ and $1\leq \ell\leq \ELL$ define the random variable:
\begin{eqnarray}\nonumber \Delta_\ell(w)&\equiv &  N_{n-1,\ell}(w)\,d_\ell(\widehat{p}_{n,\ell}(\cdot|w),\overline{p}_{n,\ell}(\cdot|w))^2 \\ \label{eq:defdeltaell}& & - [4\ln(2 + 2\log_2N_{n-1,\ell}(w)) + \ln|\sS|].\end{eqnarray}
The definition of $\conf_\ell(w)$ in Definition \ref{def:secondchoice} implies
$$\left\|\left\{\frac{d_\ell(\widehat{p}_{n,\ell}(\cdot|w),\overline{p}_{n,\ell}(\cdot|w)}{\conf_\ell(w)}\right\}_{\ell=1}^{\ELL}\right\|_{\ELL,2}>1\Leftrightarrow \frac{1}{\ELL}\sum_{\ell=1}^{\ELL}\Delta_\ell(w) > 1+ \sqrt{\frac{6\ln\left(n^2/\delta\right)}{\ELL}}.$$
Lemma \ref{lem:martingalecondensed} implies that, conditionally on $N_{n-1,\ell}(w)>0$, $\Delta_\ell(w)$ is dominated by an exponential random variable with mean $1$. The independence of the $\ELL$ processes implies that
$$\Pr{\frac{1}{\ELL}\sum_{\ell=1}^L\Delta_{\ell}(w) > 1+  \sqrt{\frac{6\ln\left(n^2/\delta\right)}{\ELL}}\mid \min_\ell N_{n-1,\ell}(w)>0}$$
can be upper bounded as if the $\Delta_\ell(w)$'s were independent exponentials. A standard Laplace transform calculation implies
$$\Pr{\frac{1}{\ELL}\sum_{\ell=1}^L\Delta_{\ell}(w) > 1+ \eps\mid \min_\ell N_{n-1,\ell}(w)>0}\leq e^{-\frac{\eps^2\ELL}{4+2\eps}}.$$
We apply this with $\eps=\sqrt{6\ln(n^2/\delta)/\ELL}$. Since $\eps\leq 1$ the RHS is $\leq \delta/n^2$. We deduce that for all $w\in A^*$
\begin{eqnarray*}\Pr{\left\|\left\{\frac{d_\ell(\widehat{p}_{n,\ell}(\cdot|w),\overline{p}_{n,\ell}(\cdot|w)}{\conf_\ell(w)}\right\}_{\ell=1}^{\ELL}\right\|_{\ELL,2}>1\mid \min_{1\leq \ell \leq \ELL} N_{n-1,\ell}(w)>0}\\ \leq \frac{\delta}{n^2}\prod_{\ell=1}^{\ELL}\Pr{N_{n-1,\ell}(w)>0}.\end{eqnarray*}
The rest of the proof follows the argument for the ``General case".\end{proof}


\section{Proof of Theorem 2}\label{Ap:betamixingtypicality}


The proof of Theorem \ref{thm:adapt} follows from the oracle inequality in Theorem \ref{thm:main} restricted to $\sT(h,\pi_*)$ and properly replacing the empirical confidence radii with the population confidence radii. Lemma \ref{thm:betamixingtypicality} stated below (proven in the Supplementary Material \cite{SM-AGCT}) establishes that the frequencies $N_{n-1,\ell}(w)$ are close to $\pi_\ell(w)n$ for typical trees provided a sample size condition holds. In turn, Lemma \ref{lem:confbeta} below allows one to switch from empirical to population confidence radii, at the price of a small multiplicative constant.

In what follows we use
$$\beta^{-1}(x)\equiv \min\{b\in \N\,:\, \forall b'\geq b,\, \beta(b')\leq x\}\,\,(x\in (0,1)).$$

\begin{restatable}{lemma}{primelemma}
\label{thm:betamixingtypicality} Let $X=(X_{k})_{k\in\Z}$ be a stationary and $\beta$-mixing process over alphabet $A$ with mixing rate function $\beta(\cdot)$. Consider a non-empty finite set $S\subset A^*$ and define:
$$h_S\equiv \max_{w\in S}|w|, \ \ \ \pi_S\equiv \min_{w\in S}\pi(w) \ \ \mbox{ where }\pi(w)\equiv \Pr{X^{-1}_{-|w|}=w}.$$
Let $\xi>0$, $\delta_0\in (0,1/e)$ and $n\in\N$ satisfy:
$$n\geq  2\,\left\{\left\lceil\frac{10h_S}{\xi}\right\rceil\vee \beta^{-1}\left(\frac{\xi\,\pi_S\,\delta_0}{24}\right)\right\}\times \left\{1 + \frac{300}{\xi^2\,\pi_S}\,\ln\left(\frac{12|S|}{\delta_0}\right)\right\},$$
then the random variables $$N_n(w)\equiv |\{|w|\leq j\leq n\,:\, X^j_{j-|w|+1}=w\}|, \ \ w \in S,$$
satisfy:
$$\Pr{\forall w\in S,\, 1-\xi\leq\frac{N_{n}(w)}{\pi(w)\,n}\leq 1+\xi}\geq 1 - \delta_0.$$
\end{restatable}

\begin{lemma}\label{lem:confbeta} Assume $X(1),\dots,X(\ELL)$ satisfy Assumptions \ref{assum:supp} through \ref{assum:polybeta}, and the sample size $n$ obeys
$$n\geq 2\,\max\left\{40h,\left\lceil\frac{48\,\Gamma\,\LL}{\pi_*\,\delta_0}\right\rceil^{1/\bgamma}\right\}\times \left\{1 + \frac{1200}{\pi_*}\,\log\left(\frac{24\,(h+1)}{\delta_0\,\pi_*}\right)\right\}.$$
 Let
$${\rm Typ}_r\equiv \bigcap_{\tilde T\in \sT(h,\pi_*)}\bigcap_{w\text{ leaf of }\tilde T}\left\{\frac{\RR{r}{\conf(w)}}{\RR{r}{{\oconf(w)}}}\leq  \sqrt{2}\right\}.$$
Then $\Pr{{\rm Typ}_r}\geq 1-\delta_0$.\end{lemma}

\begin{proof}Define a set $S$ consisting of all $w\in A^*$ of length $|w|\leq h$ and $\min_\ell\pi_\ell(w)\geq \pi_*$. This set contains all leaves of trees $T\in\sT(h,\pi_*)$, and it is clear from the definitions of confidence radii that
$$E=\bigcap_{\ell=1}^{\ELL}E_\ell\mbox{ with }E_\ell\equiv\left\{\forall w\in S:\, \frac{1}{2}\leq \frac{N_{n-1}(w)}{n\,\pi_\ell(w)}\leq \frac{3}{2}\right\}$$
is contained in ${\rm Typ}_r$.
We will apply the previous lemma to prove $\Pr{E_\ell}\geq 1-\delta_0/\ELL$, which implies $\Pr{E}\geq 1-\delta_0$ and finishes the proof. We have processes $X(1),\dots,X(\LL)$ as in Definition \ref{def:typical}, and choose parameters $n\geq 9$, $\xi=1/2$, $\delta_0=\delta_0/\LL$. The mixing rate function $\beta(b)=\Gamma\,b^{-\bgamma}$ is the same for all processes. To obtain a bound on $|S|$, we note that
$$|S\cap  A^{-1}_{-k}|\,\pi_* \leq  \sum_{w\in S\cap A^{-1}_{-k}}\Pr{X^{-1}_{-k}(1)=w}\leq \sum_{w\in S\cap A^{-1}_{-k}}\,\pi_1(w)\leq 1,$$
so
$$|S|\leq \sum_{k=0}^h\,|S\cap  A^{-1}_{-k}|\leq \frac{h+1}{\pi_*}$$
Thus we see that, in order to apply Lemma \ref{thm:betamixingtypicality} to $X(\ell)$, we need the condition
$$n\geq 2\,\max\left\{40h,\left\lceil\frac{48\,\Gamma\,\LL}{\pi_*\,\delta_0}\right\rceil^{1/\bgamma}\right\}\times \left\{1 + \frac{1200}{\pi_*}\,\log\left(\frac{24\,(h+1)}{\delta_0\,\pi_*}\right)\right\},$$
which is precisely the assumption in the present Lemma. This implies that Lemma \ref{thm:betamixingtypicality} is indeed applicable, and we deduce $\Pr{E_\ell}\geq 1-\delta_0/\ELL$, as desired.\end{proof}

\section*{Acknowledgements}
The authors would like to thank Victor Chernozhukov, Antonio Galves and Matthieu Lerasle for various discussions and to Whitney K. Newey for suggesting the dynamic choice model application.


\bibliography{biblioVLMC}
\bibliographystyle{plain}

\clearpage


\setcounter{page}{1}

\begin{center}
{\LARGE Supplementary Material for the paper ``Approximate group context tree"}
\end{center}

\setcounter{section}{0}

\section{Oracle Approximate Context Tree}\label{Sec:Oracle}

For a given sample the (minimal) compatible context tree might be too long to be efficiently estimated. Thus, in some cases, it is possible that a smaller tree, that is slightly misspecified, lead to much more efficient estimates for the conditional probabilities than any compatible tree would due to the large variance. This motivates us to consider an oracle context tree that balances bias and variance as our goal for estimation.

Based on the sample, define the ``oracle conditional probability" given a context $w$, $ \ a \in A, \  \ell=1,\ldots,\LL$ as 
$$\overline{p}_{n,\ell}(a|w)\equiv \frac{1}{N_{n-1,\ell}(w)}\sum_{i=|w|+1}^{n}\Ind{\{X^{i-1}_{i-|w|}(\ell)=w\}}\,p_\ell(a|X^{i-1}_{-\infty}(\ell)),$$
if $\min_{1\leq \ell\leq \LL}N_{n-1,\ell}(w)>0$, and  we define $\overline{p}_{n,\ell}(a|w)\equiv 1/|A|$ otherwise.

The conditional probability distribution $\overline{p}_{n,\ell}(\cdot|w)$ will play the role of an oracle estimate for the conditional probability $p_\ell(\cdot|w)$ which is adapted to the given sample. Thus $\bar p_n$ is an intermediate step in the estimation of our ultimate goal $p$. Indeed, under mild regularity conditions it follows that $\overline{p}_{n,\ell}(a|X^{-1}_{-k}(\ell))$ converges to $ p_\ell(a|X^{-1}_{-\infty}(\ell))$ as $k$ and $n$ grow at appropriate rates. For each  context $w \in A^*$, we denote the approximation error of using $\bar p_n(\cdot |w)$ as an approximation for the underlying conditional probabilities as
\begin{equation}\label{Def:cT}c_{w} := \sup_{z\in (A^{-1}_{-\infty})^\LL} \  \   \MM{\k}{d(p(\cdot|z),\bar{p}_n(\cdot|w))} \ :  \ z(\ell) \succeq w, 1\leq\ell\leq\LL,\end{equation}
where we note that $c_w \leq \MM{\k}{\gamma(w)}$.

Given a ``confidence radius" $$\conf(w)=\left(\conf_1(w),\ldots,\conf_\LL(w)\right) \ \ \mbox{for each} \ \ w\in A^*,$$ the context tree for the approximate model solves the following oracle problem for some $\k\geq 1$ and $r\geq 1$:
\begin{equation}\label{Def:Oracle} \min_{\widetilde T} \sup_{x \in A^{-1}_{-\infty}} c_{\widetilde T(x)} \ \  + \ \ \RR{r}{\conf(\widetilde T(x))}\end{equation} where the minimum is over all finite trees.

The context tree $T$ that solves the oracle problem (\ref{Def:Oracle}) balances the bias of a misspecified model and the variance associated with its estimation measured as a function of the confidence radius which is assumed to be componentwise increasing in $w$, that is, $\conf_\ell(w) \leq \conf_\ell(w')$ if $w'\succeq w$ . For convenience we also assume that $0\leq \conf_\ell(w)\leq 1$ for all $w\in A^*$.


\begin{remark}[On the oracle problem, non-uniqueness]  The oracle problem (\ref{Def:Oracle}) might have multiple solutions. Although the results derived here allow for any such solution to be considered, the oracle further selects a context tree by fixing the paths which achieve the optimal value of the oracle's objective function and further minimization of the criterion function over the remaining paths.
\end{remark}

\begin{remark}[On the oracle problem, approximation error]
Under mild conditions on the processes the oracle can adjust the length of the contexts in the oracle tree to make the approximation error $c_{T(x)}$ to be (at most) of the same order as the regularization term, namely, there is a constant $K$ such that uniformly in $x \in A^{-1}_{-\infty}$ we have
\begin{equation}\label{Def:BoundcT} c_{T(x)} \leq K \RR{r}{\conf(T(x))}. \end{equation}
\end{remark}


In addition to the oracle inequality established in Theorem \ref{thm:main} that also considers the oracle tree, it is possible to derive bounds on the actual length of $\widehat T_n(x)$ relative to $T(x)$ under the following regularity condition.

\begin{condition}[\textbf{RL}]
There is a $\kappa \in (0,1)$ and integer $\bar{k}\geq 1$ such that
{\small $$ \sup_{x \in A^{-1}_{-\infty},k} \left\{ \frac{\RR{r}{\conf(x^{-1}_{-k})}}{\RR{r}{\conf(x^{-1}_{-k-\bar{k}})}} :
|T(x)|\leq k \leq |T(x)|+\frac{\bar{k}\log (1/\RR{r}{\conf(T(x))})}{\log (1/\kappa)}
  \right\} \leq \kappa. $$}
\end{condition}

Condition RL is similar to the modulus of continuity between the regularization penalty and the length in a neighborhood of $T(x)$. 
Under Condition RL, we can establish the following result.

\begin{theorem}\label{Thm:RL}
Under Assumptions \ref{assum:supp}, \ref{assum:cont} and RL, suppose that the event $\Good_m$ occurs. Then for all $x\in A^{-1}_{-\infty}$ we have that for the oracle tree $T$
$$ | \widehat T_n(x) |  \leq  |T(x) | +\frac{\bar{k}}{\log (1/\kappa)}\max\left\{ 0 , \ \log \left(\frac{2}{c-1}\frac{c_{T(x)}}{\RR{r}{\conf(T(x))}}\right)\right\}. $$
\end{theorem}
\begin{proof}[Proof of Theorem \ref{Thm:RL}]
First we show that $ | \widehat T_n(x) |  \leq  M_x := |T(x) | +\frac{\bar{k}\log (1/\RR{r}{\conf(T(x))})}{\log (1/\kappa)}$. Suppose otherwise so that $ | \widehat T_n(x) |  >  M_x$. Then we have $$
\begin{array}{rl}
1 \geq_{(1)}  \RR{r}{\conf(\widehat T_n(x))}  & >_{(2)}  \RR{r}{\conf(T(x))} \prod_{m=0}^{(M_x-|T(x)|)/\bar k}\frac{\RR{r}{\conf(x^{-1}_{-|T(x)|-m\bar{k}})}}{\RR{r}{\conf(x^{-1}_{-|T(x)|-(m+1)\bar{k}})}} \\
  & \geq_{(3)} \RR{r}{\conf(T(x))} \prod_{m=0}^{(M_x-|T(x)|)/\bar k} \left(\frac{1}{\kappa}\right)\\
  & =_{(4)}\RR{r}{\conf(T(x))} \left(\frac{1}{\kappa}\right)^{\frac{\log (1/\RR{r}{\conf(T(x))})}{\log (1/\kappa)}}=1.\end{array}$$
where (1) follows from $\conf_\ell(\widehat T_n(x)) \leq 1$ so that $\RR{r}{\conf(\widehat T_n(x))} \leq 1$, (2) follows from $ | \widehat T_n(x) |  >  M_x$,  (3) from Condition RL, and (4) by definition of $M_x$. Therefore, $ | \widehat T_n(x) |  \leq M_x$ and Condition RL applies to $|T(x)| \leq k \leq  | \widehat T_n(x) | $.

Using the same arguments in the proofs of Claim \ref{Claim:1} with $c_{T(x)}$ instead of the continuity rates $\gamma(T(x))$, for $|\widehat T_n(x)|>|T(x)|$ we have
 \begin{equation}\label{RLub} \RR{r}{\conf(\widehat T_n(x))} \leq  \frac{2c_{T(x)}}{c-1}. \end{equation}
Next, by Condition RL, we have  \begin{equation}\label{RLlb}
\begin{array}{rl}
 \RR{r}{\conf(\widehat T_n(x))}  & =  \RR{r}{\conf(T(x))} \prod_{m=0}^{(|\widehat T_n(x)|-|T(x)|)/\bar k}\frac{\RR{r}{\conf(x^{-1}_{-|T(x)|-m\bar{k}})}}{\RR{r}{\conf(x^{-1}_{-|T(x)|-(m+1)\bar{k}})}} \\
  & \geq \RR{r}{\conf(T(x))} \left(\frac{1}{\kappa}\right)^{\frac{|\widehat T_n(x)|-|T(x)|}{\bar{k}}}.\end{array}\end{equation}   Combining (\ref{RLub}) and (\ref{RLlb}) we have
$$ \left(\frac{1}{\kappa}\right)^{\frac{|\widehat T_n(x)|-|T(x)|}{\bar{k}}} \leq \frac{2c_{T(x)}}{(c-1)\RR{r}{\conf(T(x))}}$$
and the result follows.
\end{proof}

It is interesting to consider the result of Theorem \ref{Thm:RL} when  (\ref{Def:BoundcT}) holds. In this case we have
$$ | \widehat T_n(x) |  \leq  |T(x)| + \frac{\bar{k}}{\log (1/\kappa)}\max\left\{ 0 , \ \log K +  \log\left(2/(c-1)\right)\right\}. $$
Moreover, under mild conditions on the process, $\kappa$ is bounded away from zero and $\bar{k}$ is bounded above uniformly in $n$. Therefore, these regularity conditions imply that the length of $\widehat T_n(x)$ is not larger than the length of $T(x)$ plus a constant factor.

\section{On Complete Context Trees}

\begin{remark}[Complete context trees and minimality]\label{rem:completetrees} It is sometimes convenient to work with {\em complete} trees, i.e. trees $\widetilde T$ where all non-leaf nodes have exactly $|A|$ children. For a complete tree we have that if $x,y\in A^{-1}_{-\infty}$
\begin{equation}\label{eq:completeconditionSSS}y^{-1}_{-|\widetilde T(x)|}=x^{-1}_{-|{\widetilde T}(x)|}\Rightarrow \widetilde T(x)=\widetilde T(y).\end{equation}
Completeness is an important condition when dealing with tree recovery. There exists a {\em minimal complete context tree compatible with $X(1),\dots,X(\LL)$}, defined as the intersection of all complete trees that are compatible with the $\LL$ processes. On the other hand, non-complete trees lack a well-defined minimum tree. Consider for instance a single process with $A=\{0,1\}$ and minimal complete tree
\[T^*:=T=\{e,0,1,00,10\};\]
in particular, the transition probabilities associated with the leaves $00$, $10$, and $1$ are all different. Now note that the two incomplete trees
\begin{eqnarray*}\widetilde T&:=& \{e,0,1,00\};\\
\widetilde {\widetilde T}&:=&\{e,0,1,10\},\end{eqnarray*}
are both compatible with the same process (both map $\dots 1$, $\dots 00$ and $\dots 10$ to different terminal nodes), even though there is no natural inclusion in either direction between the two trees. We also note that both trees violate condition (\ref{eq:completeconditionSSS}) above. \end{remark}

\begin{remark}[Complete context trees, continued] In general, $\widehat{T}_n$ will not be {\em complete} in the sense of Remark \ref{rem:completetrees}. Completeness can always be achieved by adding leaf nodes to non-leaf nodes of the tree created by the algorithm once it has  finished. In that case, the conditional probabilities for the added leaves are set to the conditional probabilities of their corresponding parent which was not pruned by the algorithm.\end{remark}

\section{Simulations}\label{Sec:Simulations}

In this section we conduct Monte Carlo experiments to assess the finite-sample performance of the proposed estimator. We use two different designs for the true context tree $T^*$ in these experiments: (i) a full binary Markov chain of order 3, and (ii) a sparse binary VLMC with infinite length associated with a renewal process. The associated context trees are displayed in Table \ref{Table:Models}. The former model corresponds to a parametric model with well separated conditional probabilities. The latter corresponds to an infinite context tree induced by a renewal process. These designs are two extreme cases to help illustrate the performance of the estimator on balanced and unbalanced context trees. For each design we consider the size of the group to be $\LL = 1, 10, 100$, various sample sizes $n$, and two different choices of the regularization parameter, $c=1.01, 0.5$. In all simulations we used $d_\ell=\|\cdot\|_\infty$, $k = 1$, $r=2$, $m = 2$, and  set the confidence level with $1-\delta = 0.95$. On each cell we report the relative frequency with which a particular node of the tree was selected by the proposed estimator across different repetitions. The row labelled as ``extra" corresponds to false positive, that is, the average number of nodes/suffixes not in the true  that context were selected by the algorithm. The row labelled as ``others" reports the average number of nodes in the context tree that were selected but do not belong to the list of nodes already displayed (this is relevant for the infinite trees).

Table \ref{Table:Full3binaryc05} displays the model selection performance of the proposed algorithm for the full binary Markov chain of order 3 when the parameter $c$ is set to $1.01$ and $0.5$. In the case of $c=1.01$ that follows the theoretical recommendation of the previous section, in every instance the estimated tree $\widehat T_n$ was contained in the true context tree $T^*$ confirming our theoretical results. Moreover, the estimated context tree contained a full binary Markov chain of order 2 in most instances. In the larger sample size with $100$ groups, we achieved perfect recovery of the model. When we set $c=0.5$ additional nodes not in $T^*$ are occasionally included (the average number of extra nodes is displayed in the last row of the table labeled as ``extra"). If multiple groups are used in the estimation, the number of extra nodes selected was smaller.

{\small \begin{table}[bp]
\begin{center}
\begin{tabular}{ccccc}
\pstree[levelsep=0.6cm,treesep=0.2cm]{\Toval{root}}
{\pstree{\Tcircle{0}}
        {
        \pstree{\Tcircle{0}}
        {\Tcircle{0}~{$\left(\frac{3}{4},\frac{1}{4}\right)$} \Tcircle{1}~{$\left(\frac{1}{2},\frac{1}{2}\right)$}}
        \pstree{\Tcircle{1}}
        {\Tcircle{0}~{$\left(\frac{1}{2},\frac{1}{2}\right)$} \Tcircle{1}~{$\left(\frac{1}{4},\frac{3}{4}\right)$}}
        }
\pstree{\Tcircle{1}}
        {
        \pstree{\Tcircle{0}}
        {\Tcircle{0}~{$\left(\frac{3}{4},\frac{1}{4}\right)$} \Tcircle{1}~{$\left(\frac{1}{2},\frac{1}{2}\right)$}}
        \pstree{\Tcircle{1}~*[tnpos=r]{}}
        {\Tcircle{0}~{$\left(\frac{1}{2},\frac{1}{2}\right)$} \Tcircle{1}~{$\left(\frac{1}{4},\frac{3}{4}\right)$}}
        }
}
&  &
\pstree[levelsep=0.6cm,treesep=0.35cm]{\Toval{root}}
{  \pstree{\Tcircle{0}}
        {
        \pstree{\Tcircle{0}}
            {
            \pstree{\Tcircle{0}}
                {
                \Tfan[linestyle=dashed]{}
                \Tcircle{1}
                }
            \Tcircle{1}
            }
        \Tcircle{1}
        }
   \Tcircle{1}
}
\\
\end{tabular}
\end{center}\caption{The context trees above illustrate the two models used in our simulations. The left context tree correspond to a full binary Markov chain of order 3. The right context tree correspond a process with infinite memory associated with a renewal process.}
\label{Table:Models}\end{table}
}

{\small
\begin{center}
\begin{table}[h!]
\begin{center}
\begin{tabular}{r|ccc|ccc|ccc}
     & \multicolumn{9}{c}{Probability of selection with parameter $c=1.01$}   \\
     &  \multicolumn{3}{c}{n=1000} & \multicolumn{3}{c}{n=2500} & \multicolumn{3}{c}{n=5000}  \\
\hline
Node &  \LL=1  & \LL=10  & \LL = 100 &  \LL=1  & \LL=10  & \LL = 100&  \LL=1  & \LL=10 &  \\
\hline
000  &   0.0 &  0.0  &   0.0  &  0.02  &  0.48 &   1.0  &  0.68 & 1.0\\
100  &   0.0 &  0.0  &   0.0  &  0.02  &  0.48 &   1.0  &  0.68 & 1.0\\
010  &   0.0 &  0.0  &   0.0  &  0.03  &  0.17 &   0.83 &  0.57 & 1.0\\
110  &   0.0 &  0.0  &   0.0  &  0.03  &  0.17 &   0.83 &  0.57 & 1.0\\
001  &   0.0 &  0.0  &   0.0  &  0.03  &  0.11 &   0.09 &  0.53 & 1.0\\
101  &   0.0 &  0.0  &   0.0  &  0.03  &  0.11 &   0.09 &  0.53 & 1.0\\
011  &   0.0 &  0.0  &   0.0  &  0.04  &  0.42 &   1.0  &  0.69 & 1.0\\
111  &   0.0 &  0.0  &   0.0  &  0.04  &  0.42 &   1.0  &  0.69 & 1.0\\
00   &   0.36&  1.0  &   1.0  &  1.0   &  1.0  &   1.0  &  1.0  & 1.0\\
10   &   0.36&  1.0  &   1.0  &  1.0   &  1.0  &   1.0  &  1.0  & 1.0\\
01   &   0.4 &  0.98 &   1.0  &  1.0   &  1.0  &   1.0  &  1.0  & 1.0\\
11   &   0.4 &  0.98 &   1.0  &  1.0   &  1.0  &   1.0  &  1.0  & 1.0\\
0    &   0.66&  1.0  &   1.0  &  1.0   &  1.0  &   1.0  &  1.0  & 1.0\\
1    &   0.66&  1.0  &   1.0  &  1.0   &  1.0  &   1.0  &  1.0  & 1.0\\
root &   1.0 &  1.0  &   1.0  &  1.0   &  1.0  &   1.0  &  1.0  & 1.0\\
\hline
extra &   0.0 & 0.0  &  0.0   &  0.0   &  0.0  &   0.0  &  0.0  & 0.0\\
\hline
\end{tabular}
\begin{tabular}{r|ccc|ccc|ccc}
     & \multicolumn{9}{c}{Probability of selection with parameter $c=0.5$}   \\
     & \multicolumn{3}{c}{n=1000} & \multicolumn{3}{c}{n=2500} &  \multicolumn{3}{c}{n=5000}   \\
\hline
Node  & \LL=1& \LL=10 & \LL=100 &  \LL=1  & \LL=10  & \LL=100 &  \LL=1  & \LL=10 & \\
\hline
000  & 0.92 &   1.0   &  1.0    &  1.0    &  1.0    &  1.0    &  1.0    &  1.0 \\
100  & 0.92 &   1.0   &  1.0    &  1.0    &  1.0    &  1.0    &  1.0    &  1.0 \\
010  & 0.83 &   1.0   &  1.0    &  1.0    &  1.0    &  1.0    &  1.0    &  1.0 \\
110  & 0.83 &   1.0   &  1.0    &  1.0    &  1.0    &  1.0    &  1.0    &  1.0 \\
001  & 0.84 &   1.0   &  1.0    &  1.0    &  1.0    &  1.0    &  1.0    &  1.0 \\
101  & 0.84 &   1.0   &  1.0    &  1.0    &  1.0    &  1.0    &  1.0    &  1.0 \\
011  & 0.91 &   1.0   &  1.0    &  1.0    &  1.0    &  1.0    &  1.0    &  1.0 \\
111  & 0.91 &   1.0   &  1.0    &  1.0    &  1.0    &  1.0    &  1.0    &  1.0 \\
00   & 1.0  &   1.0   &  1.0    &  1.0    &  1.0    &  1.0    &  1.0    &  1.0 \\
10   & 1.0  &   1.0   &  1.0    &  1.0    &  1.0    &  1.0    &  1.0    &  1.0 \\
01   & 1.0  &   1.0   &  1.0    &  1.0    &  1.0    &  1.0    &  1.0    &  1.0 \\
11   & 1.0  &   1.0   &  1.0    &  1.0    &  1.0    &  1.0    &  1.0    &  1.0 \\
0    & 1.0  &   1.0   &  1.0    &  1.0    &  1.0    &  1.0    &  1.0    &  1.0 \\
1    & 1.0  &   1.0   &  1.0    &  1.0    &  1.0    &  1.0    &  1.0    &  1.0 \\
root & 1.0  &   1.0   &  1.0    &  1.0    &  1.0    &  1.0    &  1.0    &  1.0 \\
\hline
extra & 4.22 &  0.0   &  0.0    &  8.67   &  0.25   &  0.0    &  27.55  &  11.02\\
\hline
\end{tabular}
\end{center}\caption{The table illustrates the model selection performance for selecting nodes of the true context tree in the full binary Markov chain of order 3.}\label{Table:Full3binaryc05}
\end{table}
\end{center}}

The renewal process is defined by the independent times $t_i$'s between observing two $1$'s. We specified the random variable $t_i$ such as $P( t_i = k ) =1/[(2\log 2 - 1)k(4k^2-1)]$. Stationarity requires that the first time be drawn from a different distribution, see \cite{CsiszarShields1996,Garivier2006} for details. We note that the polynomial decay of the tail suggests potentially long estimated trees. Table \ref{Table:Renewalbinaryc05} displays the model selection performance of the proposed algorithm for the renewal process when the parameter $c$ is set to $1.01$ and $0.5$. In the case of $c=1.01$ that follows the theoretical recommendation of the previous section, in every instance the estimated tree $\widehat T_n$ was contained in the true context tree $T^*$ confirming our theoretical results.
As expected, as the sample size increases the estimated context tree also increases chasing the infinite true context tree.
When we set $c=0.5$ additional nodes not in $T^*$ are occasionally included (the average number of extra nodes is displayed in the last row of the table labeled as ``extra"). Nonetheless, when multiple groups are used in the estimation, no node outside of the true context tree are selected.

{\small
\begin{center}
\begin{table}[h!]
\begin{center}
\begin{tabular}{r|ccc|ccc|ccc}
     & \multicolumn{9}{c}{Probability of selection with parameter $c=1.01$}   \\
     & \multicolumn{3}{c}{n=5000} & \multicolumn{3}{c}{n=10000} &  \multicolumn{3}{c}{n=50000}   \\
\hline
Node     & \LL=1& \LL=10 & \LL=100 &  \LL=1  & \LL=10  & \LL=100 &  \LL=1  & \LL=10 & \\
\hline
others  & 0.0 & 0.0  & 0.0 &  0.0   &  0.0 &  0.0 &  0.0 &   0.0   \\
00000000  & 0.0 & 0.0  & 0.0 &  0.0   &  0.0 &  0.0 &  0.01 &   0.0   \\
10000000  & 0.0 & 0.0  & 0.0 &  0.0   &  0.0 &  0.0 &  0.01 &   0.0   \\
0000000  & 0.0 & 0.0  & 0.0 &  0.0   &  0.0 &  0.0 &  0.03 &   0.0   \\
1000000  & 0.0 & 0.0  & 0.0 &  0.0   &  0.0 &  0.0 &  0.03 &   0.0   \\
000000  & 0.0 & 0.0  & 0.0 &  0.0  &  0.0 &  0.0 &  0.46 &   0.31 \\
100000  & 0.0 & 0.0  & 0.0 &  0.0  &  0.0 &  0.0 &  0.46 &   0.31 \\
00000  & 0.0 & 0.0  & 0.0 &  0.01  &  0.0 &  0.0 &  0.97 &   1.0 \\
10000  & 0.0 & 0.0  & 0.0 &  0.01  &  0.0 &  0.0 &  0.97 &   1.0 \\
0000  & 0.02 & 0.0  & 0.75 &  0.37  &  0.92 &  1.0 &  1.0 &   1.0 \\
1000  & 0.02 & 0.0  & 0.75 &  0.37  &  0.92 &  1.0 &  1.0 &   1.0 \\
000   & 0.72 & 1.0  & 1.0 &  0.99  &  1.0  &  1.0  &  1.0  &   1.0 \\
100   & 0.72 & 1.0  & 1.0 &  0.99  &  1.0  &  1.0  &  1.0  &   1.0 \\
00   & 1.0 & 1.0 & 1.0 &     1.0  &  1.0  &  1.0  &  1.0  &   1.0 \\
10   & 1.0 & 1.0 & 1.0 &     1.0  &  1.0  &  1.0  &  1.0  &   1.0 \\
0    & 1.0 & 1.0  & 1.0 &    1.0   &  1.0  &  1.0  &  1.0  &   1.0 \\
1    & 1.0 & 1.0  & 1.0 &    1.0   &  1.0  &  1.0  &  1.0  &   1.0 \\
root & 1.0  & 1.0  & 1.0 &  1.0   &  1.0  &  1.0  &  1.0  &   1.0 \\
\hline
extra & 0.0 & 0.0 & 0.0 &  0.0  &  0.0  &  0.0  &   0.0 &   0.0    \\
\hline
\end{tabular}
\begin{tabular}{r|ccc|ccc|ccc}
     & \multicolumn{9}{c}{Probability of selection with parameter $c=0.5$}   \\
     & \multicolumn{3}{c}{n=5000} & \multicolumn{3}{c}{n=10000} &  \multicolumn{3}{c}{n=50000}   \\
\hline
Node     & \LL=1& \LL=10 & \LL=100 &  \LL=1  & \LL=10  & \LL=100 &  \LL=1  & \LL=10 & \\
\hline
others    & 0.18 & 0.0  & 0.0 &  0.02   &  0.0 &  0.0 &  1.84 &   2.76   \\
00000000  & 0.03 & 0.0  & 0.0 &  0.03   &  0.0 &  0.0 &  0.77 &   1.0   \\
10000000  & 0.03 & 0.0  & 0.0 &  0.03   &  0.0 &  0.0 &  0.77 &   1.0   \\
0000000  & 0.04 & 0.0  & 0.0 &  0.06   &  0.02 &  0.0 &  0.99 &   1.0   \\
1000000  & 0.04 & 0.0  & 0.0 &  0.06   &  0.02 &  0.0 &  0.99 &   1.0   \\
000000  & 0.08 & 0.0  & 0.0 &  0.36  &  0.53 &  0.72 &  1.0 &   1.0 \\
100000  & 0.08 & 0.0  & 0.0 &  0.36  &  0.53 &  0.72 &  1.0 &   1.0 \\
00000  & 0.28 & 0.42  & 0.23 &  0.73  &  1.0 &  1.0 &  1.0 &   1.0 \\
10000  & 0.28 & 0.42  & 0.23 &  0.73  &  1.0 &  1.0 &  1.0 &   1.0 \\
0000  & 0.78 & 1.0  & 1.0 &  0.98  &  1.0 &  1.0 &  1.0 &   1.0 \\
1000  & 0.78 & 1.0  & 1.0 &  0.98  &  1.0 &  1.0 &  1.0 &   1.0 \\
000   & 1.0 & 1.0  & 1.0 &  1.0  &  1.0  &  1.0  &  1.0  &   1.0 \\
100   & 1.0 & 1.0  & 1.0 &  1.0  &  1.0  &  1.0  &  1.0  &   1.0 \\
00   & 1.0 & 1.0 & 1.0 &     1.0  &  1.0  &  1.0  &  1.0  &   1.0 \\
10   & 1.0 & 1.0 & 1.0 &     1.0  &  1.0  &  1.0  &  1.0  &   1.0 \\
0    & 1.0 & 1.0  & 1.0 &    1.0   &  1.0  &  1.0  &  1.0  &   1.0 \\
1    & 1.0 & 1.0  & 1.0 &    1.0   &  1.0  &  1.0  &  1.0  &   1.0 \\
root & 1.0  & 1.0  & 1.0 &  1.0   &  1.0  &  1.0  &  1.0  &   1.0 \\
\hline
extra & 4.36 & 0.0 & 0.0 &  6.5  &  0.0  &  0.0  &   3.54 &   0.0    \\
\hline
\end{tabular}
\end{center}\caption{The table illustrates the model selection performance for selecting nodes of the true context tree in the binary Markov chain induced by renewal process $X_t$.}\label{Table:Renewalbinaryc05}
\end{table}
\end{center}}

\section{Technical Properties of the Algorithm}\label{Sec:Algorithm}

We start by establishing technical results that follow from the definition of the pruning algorithm.

\underline{Procedure {\rm PruneTree}}\begin{figure}[h!]
\begin{algorithmic}[1]
\State $\widehat{T}_n\leftarrow E_n$. \Comment{In the beginning,
$\widehat{T}_n$ contains all visible strings} \For{each node $w$ of
$E_n$}
    \State ${\sf exam}(w)\leftarrow 0$\Comment{All nodes start out unexamined}
\EndFor \While{$\exists$ leaf $w\in \widehat{T}_n$ with ${\sf
exam}(w)=0$}\Comment{While there are unexamined leaves} \If{${\sf
CanRmv}(w)=1$} \Comment{If $w$ can be removed, remove it.} \State
$\widehat{T}_n\leftarrow \widehat{T}_n\backslash\{w\}$\EndIf \State ${\sf
exam}(w)\leftarrow 1$ \Comment{$w$ has been examined}
 \EndWhile \State {\bf return} $(\hat P_n, \widehat{T}_n)$ where $\hat P_n(\cdot\mid x) = \hat p_n(\cdot\mid \hat T_n(x))$ for all $x\in A^*$.
\end{algorithmic}
\caption{\small{The pruning algorithm.}}\label{fig:prune}\end{figure}

We begin with an alternative characterization of $\widehat{T}_n$.

\begin{proposition}\label{Prop:Smallest}
The estimated tree $\widehat{T}_n$ equals the smallest tree contained in $E_n$ which contains all $w\in E_n$ with ${\sf CanRmv}(w)=0$.\end{proposition}
\begin{proof}[Proof of Proposition \ref{Prop:Smallest}] Let $S_n$ denote the aforementioned ``smallest tree". Clearly, $S_n\subset \widehat{T}_n$, as any $w$ with ${\sf CanRmv}(w)=0$ will not be removed from $\widehat{T}_n$ in {\rm PruneTree}.
To prove that $E_n\backslash S_n\subset E_n\backslash
\widehat{T}_n$, we will use the following claim:
\begin{claim} If $v\in E_n\backslash S_n$, then
$\forall w\in E_n\,:\, w\succeq v\Rightarrow {\sf
CanRmv}(w)=1.$\end{claim}\begin{proof}[Proof of Claim]\,In contrapositive
form, if some $w\succeq v$ satisfies that ${\sf CanRmv}(w)=0$, that $w$
belongs to $S_n$ by definition, and then $v\in S_n$ because $S_n$ is
a tree and $v\preceq w$.\end{proof} One may use
induction starting from the leafs to deduce that, if $v\in E_n$ is
such that the conclusion of the Claim holds, it will be removed from
$\widehat{T}_n$ at some stage of {\rm PruneTree}. We deduce
$E_n\backslash S_n\subset E_n\backslash \widehat{T}_n$.\end{proof}

It turns out that, except for the root, any node of the estimated tree $\widehat{T}_n$ must be the closely connected with two nodes that yields substantially different probability distributions.

\begin{proposition}\label{prop:Survivors} Suppose $v\in \widehat{T}_n\backslash\{e\}$. Then there exist $w',w''\in E_n$ with $w'\succeq \parent(v)$, $w''\succeq v$ and
$$c\left[ \ \RR{r}{\conf(w')} + \RR{r}{\conf(w'')} \ \right]<\MM{\k}{ d(\hat{p}_n(\cdot|w'),\hat{p}_n(\cdot|w''))}.$$\end{proposition}
\begin{proof}[Proof of Proposition \ref{prop:Survivors}]Assume (to get a contradiction) that no such $w',w''$ exist. In that case one can easily check that ${\sf CanRmv}(w)=1$ for all $w\succeq v$. In particular, the subtree of $E_n$ obtained by removing $v$ and all of its descendants contains all $u$ with ${\sf CanRmv}(u)=0$. Proposition \ref{Prop:Smallest} then implies $v\not\in \widehat{T}_n$, which contradicts the assumptions of the present Proposition and finishes the proof.\end{proof}

Finally, the following result formally states the compatibility between the tree structure $\widehat{T}_n$ and the probability distributions $\hat{P}_n(\cdot|x)$ which follows immediately from the pruning definition.

\begin{proposition}\label{prop:Compatibility}Let $x,y\in A^{-1}_{-\infty}$ satisfy $\widehat{T}_n(x)=\widehat{T}_n(y)$. Then $\widehat{P}_n(\cdot|x)=\widehat{P}_n(\cdot|y)$.\end{proposition}

\section{Proofs of Section 5 of the main text}\label{sec:sec5proofs}

We begin by noting an asymptotic form of the condition on $h,\pi_*$ in Definition \ref{def:typical} under the assumption that $\Gamma,\bgamma$ are constant, in the setting of Section \ref{Sec:Beta}. Any sequence of numbers $h=h^{(n)}$ and $\pi_*=\pi_*^{(n)}$ satisfying
\begin{equation}\label{eq:asymptotictypical}h\,\pi_*^{-1} = \liloh{\frac{n}{\log n}}\mbox{ and }\pi^{-1}_* = \liloh{\frac{n^{\frac{\bgamma-\alpha}{\bgamma+1}}}{(\log n)^{\frac{\bgamma}{\bgamma+1}}}}\end{equation}
is valid in Definition \ref{def:typical} for all large enough $n$.

\begin{proof}[Proof of Theorem \ref{thm:parametric}] The $\ELL$ processes are polynomially $\beta$-mixing for any exponent $\bgamma>0$. Choosing $\bgamma$ sufficiently large (in terms of $\alpha$ and $\eps$) and checking (\ref{eq:asymptotictypical}) ensures that $(h,\pi_*)=(h_{T^{*}},\pi_{T^*})$  is a permissible choice for all large enough $n$. Applying Theorem \ref{thm:adapt} with $T=T^*$ gives:
\begin{equation}\label{eq:probparametric}\Pr{\begin{array}{c}\widehat{T}_n\subset T^*\mbox{ and }\forall x\in A^{-1}_{-\infty}\\ \|d(\widehat{P}_n(\cdot|x),p(\cdot|x))\|_{\LL,\k} = \bigoh{\sqrt{\frac{\log n}{\pi_{T^*}n}}}\end{array}}= 1-\bigoh{n^{-\xi}},\end{equation} with the expected improvement in the case of ``many processes".

We now argue that $\widehat{T}_n=T^*$ occurs whenever the following four conditions are met:
\begin{enumerate}
\item The event in (\ref{eq:probparametric}) holds;
\item ${\rm Typ}_r$ in Lemma \ref{lem:confbeta} holds;
\item ${\rm Good}_*$ in (\ref{eq:defgood1}) holds; and
\item $n$ is large enough.\end{enumerate} Combining (\ref{eq:probparametric}) with Lemmas \ref{lem:confbeta} and \ref{lem:probgood}, we see that the three events have joint probability $=1-\bigoh{n^{-\xi}}$ under our assumptions. Therefore, arguing that conditions $1$ to $4$ imply $\widehat{T}_n=T^*$ finishes the proof.

So assume the above conditions. They already imply $\widehat{T}_n\subset T^*$, so our goal is to show $\widehat{T}_n\supset T^*$. To this end it suffices to argue that any leaf $w\in T^*$ also belongs to $\widehat{T}_n$.

So fix a leaf $w$ of $T^*$. If $w=e$, $w\in \widehat{T}_n$ are we are done, so assume instead $w\neq e$. In this case, the definition of $d_{T^*}$ implies that there exists some other leaf $w'\in T^*$ descending from $w$'s parent with $\|d(p(\cdot|w),p(\cdot|w'))\|_{\LL,\k}\geq d_{T^*}$. Since $w,w'\in T^*$ are both leaves, the continuity rates of the two processes at the leaves are $0$, and we obtain from $\Good_*$ that
\begin{eqnarray*}\|d(\widehat{p}_n(\cdot|w),\widehat{p}_n(\cdot|w'))\|_{\LL,\k}&\geq &\|d(p(\cdot|w),p(\cdot|w'))\|_{\LL,\k}-\bigoh{\sqrt{\frac{\log n}{\pi_{T^*}n}}}\\ &>&d_{T^*} -\bigoh{\sqrt{\frac{\log n}{\pi_{T^*}n}}}.\end{eqnarray*}
Moreover, under ${\rm Typ}_r$,  the confidence radii $\|\conf(w)\|_{\ELL,r}$ and $\|\conf(w')\|_{\ELL,r}$ are of the order of their population counterparts $\|\oconf(w)\|_{\ELL,r},\|\oconf(w')\|_{\ELL,r}$ (respectively). Therefore,
$$c\,[\|\conf(w)\|_{\ELL,r} + \|\conf(w')\|_{\ELL,r}] = \bigoh{\sqrt{\frac{\log n}{\pi_{T^*}n}}}.$$
We have assumed that $d_T^{-1}=\liloh{\sqrt{\log n/\pi_{T^*}\,n}}$, so we conclude \[c\,[\|\conf(w)\|_{\ELL,r} + \|\conf(w')\|_{\ELL,r}]<\|d(\widehat{p}_n(\cdot|w),\widehat{p}_n(\cdot|w'))\|_{\LL,\k}\] for $n$ large enough. This implies ${\sf CanRmv}(w)=0$ and $w\in \widehat{T}_n$, as desired.\end{proof}

\begin{remark}\label{rem:BW} B\"{u}hlman and Wyner \cite{BuhlmannWyner1999} prove a related result about perfect tree recovery in the single process case ($\ELL=1$). In their assumptions they require that all leaves of $T^*$ have positive probability and also that a condition implying geometric $\phi$-mixing is satisfied. They also require
$$\pi_{T^*}^{-1} = \bigoh{\frac{\sqrt{n}}{\log^{1/2+a}n}}\mbox{ and }d^{-1}_{T^*}=\bigoh{\frac{\sqrt{\pi_{T^*} n}}{\log^{1/2+b}n}}$$
where $a,b>0$ are constants. Our estimator also satisfies this: just notice that, for complete trees, we have the crude bounds
$$\pi_{T^*}\times (\mbox{ \# of leaves of $T^*$})\leq \sum_{w\text{ leaf of }T^*}\pi_1(w)=1,\mbox{ and }$$
$$h_{T^*}\leq\mbox{(\# of leaves of $T^*$)}.$$\end{remark}

\begin{proof}[Proof of Theorem \ref{thm:complete}] We will again apply Theorem \ref{thm:adapt}. Under our assumptions the $\ELL$ processes are $\beta$-mixing with rate function $\beta(b)=\Gamma\,k^{-\bgamma}$, where $\Gamma$ depends on $\Gamma_0$, $\bgamma$ and the size of the alphabet \cite{CometsFernandezFerrari2002}. One may compare the AGCT estimator to a tree $T$ obtained by truncating the infinite $|A|$-ary tree at height $h_T=c\log n$, where $c>0$ is a constant. Non-nullness implies that $\pi_{\ell}(w) \geq \eta^{|w|}$ for any $w$, so the tree $T$ satisfies $\pi_T\geq \eta^{c\log n}$ and $(h,\pi_*)=(h_T,\pi_T)$ is a valid choice in Definition \ref{def:typical} whenever $c$ is sufficiently small. Moreover, for sufficiently small $c$ the error bound from Theorem \ref{thm:adapt} is dominated by the first term, which is $\bigoh{h_T^{-\bgamma}}$ in our case. \end{proof}

\begin{proof}[Proof of Theorem \ref{thm:renewal}]We note some basic facts. A stationary binary renewal process takes values in the alphabet $A=\{0,1\}$ and is defined by a probability distribution $\mu$ on $\N$ with finite first moment ($\mu$ is called the arrival distribution). The distribution of the corresponding process $X_0,\dots,X_k$ consists of placing ones precisely at positions
$$\tau_0,\tau_0+\tau_1,\tau_0+\tau_1+\tau_2,\tau_0+\tau_1+\tau_2+\tau_3,\dots,$$ where $\tau_0,\tau_1,\tau_2,\tau_3,\ldots$ are independent and $\tau_1, \tau_2, \dots$ have law $\mu$ ($\tau_0$ must have a different law in order to make the process stationary). We write $\mu(\geq k)\equiv \sum_{i\geq k}\mu(k)$ in what follows. One can show that:
\begin{enumerate}
\item Any binary renewal process with arrival distribution $\mu_\ell$ is compatible with the infinite tree $T^*$ with leaves $10^{k-1}$, $k\in\N_*=\N\setminus\{0\}$. Moreover $$\Pr{X^{-1}_{-k}(\ell)=10^{k-1}} = \frac{\mu_\ell(\geq k)}{\sum_{i\in \N_*}\mu_\ell(\geq i)}\geq \alpha\,\mu_\ell(\geq k)$$ for some $\alpha>0$ depending only on $C$ and $\bgamma$, since
$$\sum_{i\in \N_*} \mu(\geq i) = \sum_{k\in \N_*} \mu(k)\,k\leq \left(\sum_{k\in \N_*} \mu(k)\,k^{2+\bgamma}\right)^{\frac{1}{2+\bgamma}}\leq C^{\frac{1}{2+\bgamma}}.$$
\item The transition probabilities are $$p_\ell(1\mid 10^{k-1})=\frac{\mu_\ell(k)}{\mu_\ell(\geq k)}, \ \ \ell = 1,\ldots,\LL.$$
In particular, (\ref{eq:renewalcont}) implies that the processes $X(\ell)$ have continuous transition probabilities. The fact that the $\mu_\ell$ have full support implies that the supports of the processes are all equal to $A^{-1}_{-\infty}$. Moreover,
\begin{equation}\label{eq:momentboundrenewal}\mu_\ell(\geq k)\leq \frac{1}{k^{2+\bgamma}}\sum_{i\geq k}\mu(i)\, i^{2+\bgamma}\leq \frac{C}{k^{2+\bgamma}}.\end{equation}
\item The bound on the ($2+\bgamma$)-th moment for $\mu_\ell$ implies polynomial $\beta$-mixing with rate function $\beta(b)=\Gamma\,b^{-\bgamma}$, where $\Gamma$ depends only on $C$ and $\bgamma$ (this follows from inspecting the arguments of Section 2 in Lindvall \cite{Lindvall1979}).
\end{enumerate}
It follows that the processes satisfy the assumptions of Theorem \ref{thm:adapt}. We take $T$ to be the truncation of $T^*$ at level $h=h_T$, where $h$ is the largest positive integer such that $$\alpha\,\mu_\ell(\geq h)\geq \pi_*\equiv (\log n) \, n^{-\bgamma/(\bgamma+1)}.$$ A calculation using (\ref{eq:momentboundrenewal}) shows that $(h,\pi_*)$ is permissible in Definition \ref{def:typical}. Moreover, for all $x\in G$ we have $p_\ell(\cdot|T(x))=p_\ell(\cdot|x)$. Applying the Corollary \ref{cor:oracle1} gives the desired result.
\end{proof}

\section{Proofs of Section 6}

\begin{lemma}\label{lemma:DP}
In the discrete stochastic dynamic programming problem described above, for $q\geq 1$ the estimator $\widehat V$ satisfies
$$ \max_{x \in A^{-1}_{-\infty}} \frac{| \widehat V(x) - V(x) |}{\displaystyle\|V( x \cdot)\|_{\frac{q}{q-1}}\max_{u\in \mathcal{U}} \| \widehat P_{n,u}(\cdot|x)-p_u(\cdot|x)\|_q } \leq \frac{\disc}{1-\disc}$$ where $q/(q-1)=\infty$ if $q=1$.
\end{lemma}
\begin{proof}[Proof of Lemma \ref{lemma:DP}]
For $x \in A^{-1}_{-\infty}$ and $u\in \mathcal{U}$, denote by $V_x = V(x)$ the true value function, $P_{u,x}$ denote the true transition probability (infinite) matrix, $\widehat V_x = \widehat V(\widehat T_n(x))$ the value function associated with the estimated transition probabilities $\widehat P_{n,u,x}=\widehat P_{n,u}(\cdot\mid x)=\widehat p_{n,u}(\cdot\mid \widehat T_n(x))$.

Each value function is the fixed point of contraction mappings $H$, $\widehat H$ on the functions $W:A^{-1}_{-\infty}\to \mathbb{R}$. Formally, the mappings
\begin{eqnarray*}H( W )(x) &=& \max_{u\in \mathcal{U}} \{f(x_{-1},u) + \disc P_{u,x} W \} \mbox{ and} \\  \widehat H( W )(x) &=& \max_{u\in \mathcal{U}} \{f(x_{-1},u) + \disc \widehat P_{n,u,x} W \}\end{eqnarray*} are contractions with modulus $\disc$ by Blackwell's sufficient conditions.

Therefore we have
$$ \begin{array}{rl}
\| \widehat V - V \|_\infty & \leq \| \widehat V - \widehat H(V) \|_\infty + \| \widehat H(V) - V \|_\infty \\
& = \| \widehat H(\widehat V) - \widehat H(V) \|_\infty + \| \widehat H(V) - H(V) \|_\infty  \\
& \leq \disc  \| \widehat V - V \|_\infty  + \| \widehat H(V) - H(V) \|_\infty.  \end{array} $$
Thus, $\| \widehat V - V \|_\infty \leq \| \widehat H(V) - H(V) \|_\infty /(1-\disc)$. where $\| \widehat H(V) - H(V) \|_\infty = \max_{x\in A^{-1}_{-\infty}} | H(V)(x) - \widehat H(V)(x)|$. Thus the result follows by showing that
$$ \begin{array}{ll}& | H(V)(x) - \widehat H(V)(x)|  \\ =& \left| \max_{u\in \mathcal{U}}\{ f(x_{-1},u) + \disc\widehat P_{n,u,x} V\} - \max_{\tilde u \in \mathcal{U}} \{ f(x_{-1},\tilde u) + \disc P_{\tilde u,x} V\}\right| \\
 \leq &  \displaystyle \disc \max_{u\in \mathcal{U}}| ( \widehat P_{n,u,x} -  P_{u,x}) V | \\
 = &\disc \max_{u\in \mathcal{U}} \left|\sum_{a\in A} [ \hat p_{n,u}(a|\widehat T_n(x)) - p_u(a|x) ]V(ax)\right|\\
\leq &   \disc \|V(\cdot x)\|_{\frac{q}{q-1}} \max_{u\in \mathcal{U}} \|\widehat P_{n,u}(\cdot|x) - p_u(\cdot|x)\|_q.
 \end{array}$$
\end{proof}

\begin{proof}[Proof of Theorem \ref{thm:DP}]
The result follows by applying Lemma \ref{lemma:DP} with $q=1$ (which corresponds to $d_\ell=\|\cdot\|_1/2$) and Theorem \ref{thm:main} for the general case with $r=\k=\infty$ to bound
$$ \max_{u\in\mathcal{U}}\frac{\| \widehat P_{n,u}(\cdot|x)-p_u(\cdot|x)\|_1}{2} \leq \inf_T \frac{2c+2}{c-1}\MM{\k}{\gamma(T(x))} + (1+2c)\RR{\infty}{\conf(T(x))}.$$

The bound on  $\RR{\infty}{\conf(T(x))}$ follows from Definition \ref{def:firstchoice} with the family of sets $\sS = 2^A$.
\end{proof}

\begin{proof}[Proof of Theorem \ref{Thm:DCM}]
By the choice of $\conf$ we have that with probability at least $1-\delta$ the event $\Good_m$ occurs. Fix $a\in A$, $x,y\in A^{-1}_{-\infty}$ let
$$ \bar m_\ell(a,x,y) = \bar p_{n,\ell}(a|T(x))- \bar p_{n,\ell}(a|T(y)), \ \ \overline{{\rm AVEm}}(a,x,y)= \Ex{\bar{m}_\ell(a,x,y)},$$ and note that
\begin{equation}\label{Eq:AVEm}
\begin{array}{rl}
\displaystyle \widehat{{\rm AVEm}}(a,x,y) - {\rm AVEm}(a,x,y) &\displaystyle  = \frac{1}{\LL}\sum_{\ell=1}^\LL [\hat m_\ell(a,x,y) - \bar m_\ell(a,x,y)] + \\
&\displaystyle  + \frac{1}{\LL}\sum_{\ell=1}^\LL \bar m_\ell (a,x,y) - \overline{{\rm AVEm}}(a,x,y)+\\
& + \overline{{\rm AVEm}}(a,x,y)- {\rm AVEm}(a,x,y).\\
\end{array}
\end{equation}
First note that for $d_\ell=\|\cdot\|_\infty$, we have
{ $$\begin{array}{rl}
& \left|\overline{{\rm AVEm}}(a,x,y)- {\rm AVEm}(a,x,y)\right|\\
&\leq \MM{1}{d(\bar{p}_{n}(\cdot|T(x)), p(\cdot|x))} +\MM{1}{d( \bar{p}_{n}(\cdot|T(y)),p(\cdot|y))}\\
& \leq \MM{1}{\gamma(T(x))} + \MM{1}{\gamma(T(y))}.
\end{array}$$}

Next, define  $E_\LL(a,x,y) = \left|\frac{1}{\LL}\sum_{\ell=1}^\LL \bar m_\ell (a,x,y) - \overline{{\rm AVEm}}(a,x,y)\right|$ and note that since $|\bar m_\ell (a,w,w')| \leq 1$ we have
$$ E_\LL(a,x,y) \leq [2/\LL] + [(\LL-1)/\LL] E_{\LL-1}(a,x,y).$$

Moreover, note that $E_\LL(a,x,y) = E_\LL(a,T(x),T(y))$. Thus, we need to consider only suffixes $w \in T$ (in particular, leaves of $T$) which implies that $N_{n,\ell}(w)>0$ for every $\ell =1,\ldots,\LL$. Thus, for $\xi = (\epsilon-[2/\LL])\LL/(\LL-1)$ we have
{\small $$
\begin{array}{rl}
 \displaystyle  P \left( \max_{a\in A, w,w' \in T}E_\LL(a,w,w') \geq \epsilon \right) & \displaystyle  \leq  P \left( \max_{a\in A, w,w' \in T}E_{\LL-1}(a,w,w') \geq \xi \right)\\
&\displaystyle \leq P \left( \max_{a\in A, { w : \min_\ell N_{n,\ell}(w)>0,\atop w': \min_\ell N_{n,\ell}(w')>0}}E_{\LL-1}(a,w,w') \geq \xi \right)\\
&\displaystyle \leq P \left( \max_{a\in A, { w : N_{n,\LL}(w)>0,\atop w': N_{n,\LL}(w')>0}}E_{\LL-1}(a,w,w') \geq \xi \right)\\
&\displaystyle \leq \sum_{{ a\in A,\atop { w\in A^*,\atop w'\in A^*}}} P \left( \begin{array}{c} N_{n,\LL}(w)>0, N_{n,\LL}(w')>0,\\ E_{\LL-1}(a,w,w') \geq \xi \end{array}\right).\\
\end{array}
$$}
The event $\{ N_{n,\LL}(w)>0, N_{n,\LL}(w')>0 \}$ is independent of $\{E_{\LL-1}(a,w,w') \geq \xi\}$. Moreover, since $|\bar m_\ell (a,w,w')| \leq 1$ are i.i.d. draws from the population of agents, for any $\xi>0$
$$ P\left( E_{\LL-1}(a,w,w') > \xi \right) \leq \exp(-(\LL-1)\xi^2/2)$$
by Hoeffding's inequality. Therefore,
{ $$
\begin{array}{rl}
&\displaystyle  P \left( \max_{a\in A, w,w' \in T}E_\LL(a,w,w') \geq \epsilon \right) \\
&\displaystyle \leq |A|\exp(-(\LL-1)\xi^2/2) \cdot \sum_{{ w\in A^*\atop w'\in A^*}} P \left( N_{n,\LL}(w)>0, N_{n,\LL}(w')>0  \right).\\
\end{array}
$$}
Next note that $\sum_{{ w\in A^*\atop w'\in A^*}} P \left( N_{n,\LL}(w)>0, N_{n,\LL}(w')>0  \right)$ is the expected number of different pairs $w,w'\in A^*$ appearing in $X_1^n(\ell)$. Therefore we have
$$\sum_{{ w\in A^*\atop w'\in A^*}} P \left( N_{n,\LL}(w)>0, N_{n,\LL}(w')>0  \right) \leq n^4/4.$$
Finally, to control the first term of (\ref{Eq:AVEm}), since $d_\ell = \|\cdot\|_\infty$, we have
{\small $$
\begin{array}{rl}
\left|\frac{1}{\LL}\sum_{\ell=1}^\LL [\hat m_\ell(a,x,y) - \bar{m}_\ell(a,x,y)] \right|&  \leq \left|\frac{1}{\LL}\sum_{\ell=1}^\LL [\hat p_{n,\ell}(a|\widehat T_n(x)) - \bar{p}_{n,\ell}(a|T(x))]\right| + \\ & + \left| \frac{1}{\LL}\sum_{\ell=1}^\LL [\hat p_{n,\ell}(a|\widehat T_n(y)) - \bar p_{n,\ell}(a|T(y))]\right|\\
& \leq \MM{1}{d(\hat p_{n}(\cdot|\widehat T_n(x)), \bar p_{n}(\cdot|T(x)))} + \\ & +\MM{1}{d(\hat p_{n}(\cdot|\widehat T_n(y)), \bar p_{n}(\cdot|T(y)))}. \\
\end{array}
$$}

Under the event $\Good_2$, by Theorem \ref{thm:main}, uniformly over $z\in A^{-1}_{-\infty}$ we have that 
$$\begin{array}{rl} \MM{1}{d(\hat p_{n}(\cdot|\widehat T_n(z)), \bar p_{n}(\cdot|T(z)))} & \leq \frac{2c+2}{c-1}\MM{1}{\gamma(T(z))}+(1+2c)\RR{2}{\conf(T(z))}\end{array}$$
The result follows by combining these bounds.

\end{proof}

\section{A compendium of martingale results}\label{App:Martingale}

\begin{lemma}\label{thm:economicalazuma}Let $(M_i,\sF_i)_{i=0}^m$ be a martingale with $M_0=0$ and $|M_{i}-M_{i-1}|\leq Y_{i-1}$ for some $\sF_{i-1}$-measurable r.v. $Y_{i-1}$. Define $V_n\equiv \sum_{j=0}^{n-1}Y^2_j$. Then:
$$\forall \lambda,v>0\,:\, \Pr{M_n\geq \lambda,0<V_n\leq
v}\leq
\Pr{V_n>0}\,e^{-\frac{\lambda^2}{2v}}.$$\end{lemma}
\begin{proof} \textit{Step 1: Main arguments.} Write $E\equiv \{\sum_{j=0}^{n-1}Y_j^2>0\}$ and define
$$U_r\equiv e^{sM_r -
\frac{s^2}{2}\sum_{j=0}^{r-1}Y_j^2}\,(0\leq r\leq n)$$
where $s>0$ will be fixed later. By Step 2 below $(U_r,\sF_r)$ is a supermartingale.

Now notice that: $$M_n\geq \lambda,\sum_{j=0}^{n-1}Y_j^2\leq v\Rightarrow sM_n -s\lambda + \frac{s^2v}{2}-
\frac{s^2}{2}\sum_{j=0}^{n-1}Y_j^2 \geq 0\Rightarrow U_n e^{\frac{s^2v}{2}-s\lambda}\geq 1.$$
Therefore,
$$\Pr{M_n\geq \lambda,0<\sum_{j=0}^{n-1}Y_j^2\leq
v}\leq e^{\frac{s^2v}{2}-s\lambda}\,\Ex{U_n\,\Ind{E}}.$$
The result follows by considering $s=\lambda/v$ and noting that $\Ex{U_n\,\Ind{E}} \leq \Pr{V_n>0}$ by Step 3 below.

\textit{ Step 2: $(U_r,\sF_r)$ is a supermartingale.} Since $Y_r$ is $\sF_r$-measurable,
$$\Ex{\frac{U_{r+1}}{U_{r}}\mid\sF_r} = \Ex{e^{s(M_{r+1}-M_r)}\mid\sF_r} e^{- \frac{s^2Y^2_{r}}{2}}.$$ Recall that $|M_{r+1}-M_r|\leq Y_r$, hence by convexity
$$e^{s(M_{r+1}-M_{r})}\leq \cosh(sY_r) +\sinh(sY_r)(M_{r+1}-M_r).$$
Taking conditional expectations, we see that:
\begin{eqnarray*}\Ex{e^{s(M_{r+1}-M_{r})}\mid\sF_r}&\leq &\cosh(sY_r) + \sinh(sY_r)\,\Ex{(M_{r+1}-M_r)\mid\sF_r} \\ &=& \cosh(sY_r).\end{eqnarray*}
This implies $\Ex{e^{s(M_{r+1}-M_r)}\mid\sF_r} e^{- \frac{s^2Y^2_{r}}{2}}\leq \cosh(sY_r)e^{- \frac{s^2Y^2_{r}}{2}}\leq 1$ via the classical inequality $$\forall x\in\R\,:\,\cosh(x)\leq e^{x^2/2},$$ which directly follows from comparing Taylor expansions.

\textit{Step 3: $\Ex{U_n\,\Ind{E}} \leq \Pr{V_n>0}$.} Write $E_0\equiv
\{Y_0\neq 0\}$ and $E_{j}\equiv \{Y_j\neq 0\}\cap \{Y_k=0, 0\leq
k<j\}$. Notice that $E = \cup_{0\leq j\leq n-1}E_j$ (where the union
is disjoint) and that each $E_j$ is $\sF_j$-measurable. Moreover, if
$E_k$ holds, we have $\sum_{j=0}^{k-1}Y^2_j=0$ and $M_{k}=M_0=0$, hence $U_k=1$.
Therefore,
\begin{eqnarray*}\Ex{U_n\,\Ind{E}
} &=& \sum_{k=0}^{n-1}\Ex{U_n\,\Ind{E_k}}\\ \mbox{($E_k$ is
$\sF_k$-measurable)}&=& \sum_{k=0}^{n-1}\Ex{\Ind{E_k}\Ex{U_n\mid\sF_k}}\\
\mbox{($U_k$ is supermartingale)}&\leq & \sum_{k=0}^{n-1}\Ex{\Ind{E_k}U_k}\\
\mbox{($U_k=1$ in $E_k$)}&=& \sum_{k}\Pr{E_k}=\Pr{E}.\end{eqnarray*}
\end{proof}

\begin{lemma}\label{thm:economicalazumaSecond}Let $(M_i,\sF_i)_{i=0}^m$ be a martingale with $M_0=0$, $|M_{i}-M_{i-1}|\leq Y_{i-1} \leq 1$ for some $\sF_{i-1}$-measurable r.v. $Y_{i-1}$. Define $$\widetilde V_n\equiv \sum_{j=1}^{n}\Ex{(M_j-M_{j-1})^2|\sF_{j-1}}.$$ Then:
$$\forall \lambda,v>0\,:\, \Pr{M_n\geq \lambda,0<\widetilde V_n\leq
v}\leq
\Pr{\widetilde V_n>0}\,e^{-\frac{\lambda^2}{2v}(2-\exp(\lambda/v))}.$$\end{lemma}
\begin{proof} \textit{Step 1: Main Arguments.} Write $E\equiv \{\widetilde V_n>0\}$ and define
$$U_r\equiv e^{sM_r -
\frac{s^2e^s}{2}\sum_{j=1}^{r}E[(M_j-M_{j-1})^2|\sF_{j-1}]}\,(0\leq r\leq n)$$
where $s>0$ will be fixed later. It follows that $(U_r,\sF_r)$ is a supermartingale by Step 2 below.

Now notice that: $$M_n\geq \lambda,\widetilde V_n\leq v\Rightarrow sM_n -s\lambda + \frac{s^2e^sv}{2}-
\frac{s^2e^s}{2}\widetilde V_n \geq 0\Rightarrow U_n e^{\frac{s^2e^sv}{2}-s\lambda}\geq 1.$$
Therefore,
$$\begin{array}{rl}
\Pr{M_n\geq \lambda,0<\widetilde V_n\leq
v} & = \Ex{1\{ U_n e^{\frac{s^2e^sv}{2}-s\lambda}\geq 1 \} \cdot \Ind{E}}\\
&\leq \Ex{ U_n e^{\frac{s^2e^sv}{2}-s\lambda}  \Ind{E} } \\
& = e^{\frac{s^2e^sv}{2}-s\lambda}\,\Ex{U_n\,\Ind{E}}.\end{array}$$
The Lemma follows from choosing $s=\lambda/v$ and noting that $\Ex{U_n\,\Ind{E}}\leq \Pr{E}$ by Step 3.

\textit{Step 2: $(U_r,\sF_r)$ is a supermartingale.} Since $Y_r$ is $\sF_r$-measurable,
$$\Ex{\frac{U_{r+1}}{U_{r}}\mid\sF_r} = \Ex{e^{s(M_{r+1}-M_r)}\mid\sF_r} e^{- \frac{s^2e^sE[(M_{r+1}-M_{r})^2|\sF_{r}]}{2}}.$$ Recall that $|M_{r+1}-M_r|\leq Y_r\leq 1$, hence by \cite{LedouxTalagrandBook} page 32,
$$\Ex{e^{s(M_{r+1}-M_{r})}\mid\sF_r}\leq \exp\left( \frac{s^2e^s}{2} E[(M_{r+1}-M_{r})^2|\sF_{r}] \right).$$
This implies $\Ex{e^{s(M_{r+1}-M_r)}\mid\sF_r} e^{- \frac{s^2e^sE[(M_{r+1}-M_{r})^2|\sF_{r}]}{2}}\leq  1$.

\textit{Step 3: $\Ex{U_n\,\Ind{E}
}\leq \Pr{E}$}. The proof is similar to Step 3 in Lemma \ref{thm:economicalazuma}.
\end{proof}

For any $\gamma > 1$, $\delta\in(0,1)$, define monotonic function $h:[0,\infty) \to [0,\infty)$    $$h(x) = 2x\gamma^2\log\left\{\frac{2}{\delta}(1+\log_{\gamma}x)(2+\log_{\gamma}x)\right\},$$
and let $i_0$ be any integer such that
$$ \gamma^{i_0}\log^2(2-(1/\gamma)) \geq 2\log(2/\delta) + 2\log[(1+i_0)(2+i_0)],$$
so that $2-\exp(\sqrt{h(\gamma^{i_0})}/\gamma^{i_0+1}) \geq 1/\gamma$.

\begin{lemma}\label{thm:economicalmgfestimate} Let $(M_i,\sF_i)_{i=0}^m$ be a martingale with
 $M_0=0$, $|M_{i}-M_{i-1}|\leq Y_{i-1} \leq 1$ for some $\sF_{i-1}$-measurable binary r.v. $Y_{i-1}$. Define $\widetilde V_n\equiv \sum_{j=1}^{n}E[(M_j-M_{j-1})^2|\sF_{j-1}]$ and $V_n\equiv \sum_{j=0}^{n-1}Y_j^2$. Then
{\small $$ \Pr{M_n \geq \sqrt{h(\widetilde V_n)}, \widetilde V_n>\gamma^{i_0}} + \Pr{M_n \geq \sqrt{h(V_n)/\gamma}, 0 < \widetilde V_n \leq \gamma^{i_0} }  \leq
\delta\cdot \Pr{V_n>0}. $$}
\end{lemma}\begin{proof}
We bound the first term as {\tiny $$
\begin{array}{rl}
 \displaystyle \Pr{ M_n \geq \sqrt{h(\widetilde V_n)}, \widetilde V_n > \gamma^{i_0} } & \displaystyle  =  \sum_{i\geq i_0} \Pr{ M_n \geq \sqrt{h(\widetilde V_n)}, \gamma^i\leq \widetilde V_n  \leq \gamma^{i+1} } \\
 & \displaystyle  \leq  \sum_{i\geq i_0} \Pr{ M_n \geq  \sqrt{h(\gamma^{i})}, 0 < \widetilde V_n  \leq \gamma^{i+1}} \\
 &\displaystyle \leq  \Pr{\widetilde V_n>0}\sum_{i\geq i_0} \exp\left( - \frac{h(\gamma^{i})}{2\gamma^{i+1}}\left[2 - \exp\left(\frac{\sqrt{h(\gamma^i)}}{\gamma^{i+1}} \right)\right]\right)\\
 &\displaystyle \leq \Pr{\widetilde V_n>0}\sum_{i\geq i_0} \exp\left( -\log\frac{2}{\delta} - \log[(1+i)(2+i)]\right) \\
\end{array}
$$}where the second inequality follows by applying Lemma \ref{thm:economicalazumaSecond} for each $i\geq i_0$, and the last line follows from the definition of $i_0$.
Since $i_0 \geq 0$ it follows that
$$ \sum_{i\geq i_0} \exp\left( -\log\frac{2}{\delta} -\log[(1+i)(2+i)]\right) = \frac{\delta}{2}\sum_{i\geq i_0} \frac{1}{(1+i)(2+i)} \leq \frac{\delta}{2}.$$
Furthermore, note that $\widetilde V_n \leq V_n$ so that $\Pr{\widetilde V_n>0}\leq \Pr{V_n>0}$.

Next we proceed for the second term. Note that $\{0 < \widetilde V_n < \gamma^{i_0}\} \subseteq \{ V_n > 0\}$ and that $V_n$ only takes integer values so that
$$\{V_n>0\} = \bigcup_{i=0}^{+\infty}\{\gamma^i\leq V_n<\gamma^{i+1}\}$$
and the union is disjoint. We deduce
{\tiny $$
\begin{array}{rl}
\displaystyle  \Pr{ M_n\geq \sqrt{h(V_n)/\gamma}, 0 < \widetilde V_n \leq \gamma^{i_0} } & \leq  \displaystyle  \Pr{M_n\geq \sqrt{h(V_n)/\gamma}, V_n > 0 }\\
 & \displaystyle  =  \sum_{i\geq 0}  \Pr{ M_n\geq \sqrt{h(V_n)/\gamma}, \gamma^i\leq V_n  \leq \gamma^{i+1}} \\
 & \displaystyle  \leq  \sum_{i\geq 0}  \Pr{ M_n \geq  \sqrt{h(\gamma^{i})/\gamma}, 0 < V_n  \leq \gamma^{i+1}} \\
 &\displaystyle \leq   \Pr{V_n>0}\sum_{i\geq 0} \exp\left( - \frac{h(\gamma^{i})}{2\gamma^{i+2}}\right)\\
 &\displaystyle \leq  \Pr{V_n>0}\sum_{i\geq 0} \exp\left( -\log\frac{2}{\delta} - \log[(1+i)(2+i)]\right) \\
 & \displaystyle \leq  \Pr{V_n>0}\delta /2.
\end{array}
$$}
where we applied Lemma \ref{thm:economicalazuma} for each $i\geq 0$, and the definition of $h$.

\end{proof}

\section{Typicality results for $\beta$-mixing processes}

In what follows we use
$$\beta^{-1}(x)\equiv \min\{b\in \N\,:\, \forall b'\geq b,\, \beta(b')\leq x\}\,\,(x\in (0,1)).$$

\primelemma*

\begin{proof}[Proof of Lemma \ref{thm:betamixingtypicality}] Consider a number $b\in\N\backslash\{0\}$. Given $r\in \N$, a sequence $B=(B_1,\dots,B_{r})\in [n]$ of subsets of $[n]$ is said to consist of $b$-separated blocks if each $B_i$ is an interval in $[n]$
 and $\min B_{i+1}\geq \max B_{i}+b$ for $1\leq i\leq r-1$. We say that such a sequence is {\em $t$-regular} if $t=|B_1|=|B_2|=\dots=|B_{r-1}|\geq |B_r|$.

\begin{lemma}\label{lem:blocks}Under the assumptions of Lemma \ref{thm:betamixingtypicality}, let $(B_1,\dots,B_{r})$ be a sequence of $b$-separated $t$-regular blocks where $t\geq 2|w|$. Define for each $w\in S$ the number of occurrences of $w$ that are contained in one of the blocks $B_i$:
$$N(w)\equiv |\{j\in [n]\,:\, \exists i\in [r],\, j,j+|w|-1\in B_i\mbox{ and }X^{j+|w|-1}_{j}=w\}|$$
and let $n_w$ denote the maximum number of places where $w$ may occur: $$n_w\equiv \sum_{i=1}^r(|B_i|-|w|+1)_+ = (t-|w|+1)\,(r-1) + (|B_r|-|w|+1)_+.$$ Given $\lambda>0$, let $E(\lambda)$ denote the event:
$$E(\lambda)\equiv \left\{\forall w\in S,\, \left|\frac{N(w)}{n_w}-\pi(w)\right|\leq \lambda\,\pi(w)\right\}$$
Then
$$\Pr{E(\lambda)}\geq 1 - 2|S|\exp\left(-\frac{\lambda^2\,\pi_S\,(r-1)}{16(1 + \frac{\lambda}{6})}\right) - \frac{2\beta(b)}{\lambda\,\pi_S}.$$\end{lemma}

\begin{proof}[Proof of Lemma \ref{lem:blocks}] Let $\tilde{X}_{B_1},\dots,\tilde{X}_{B_{r+1}}$ be a sequence of independent random variables where each $\tilde{X}_{B_i}$ has the same distribution as $X_{B_i}$. Define $\tilde{N}(w)$ and $\tilde{E}(\cdot)$ in analogy with $N(w)$ and $E(\cdot)$ (respectively).

Our first major goal in the proof is to show:

\begin{claim}$\Pr{E(\lambda)}\geq \Pr{\tilde{E}(\lambda/2)} - \frac{2\beta(b)}{\pi_S\,\lambda}.$\end{claim}
To prove the claim we first construct a coupling of $\tilde{X}_{B_1},\dots,\tilde{X}_{B_{r}}$ to the process $X^{+\infty}_{-\infty}$. Set $\tilde{X}_{B_1}=X_{B_1}$. Assuming that we have defined $\tilde{X}_{B_i}$ for $1\leq i<j$, we sample $(X_{B_j},\tilde{X}_{B_j})$ from a coupling achieving total variation distance. That is to say,
$$\begin{array}{l}\Pr{X_{B_j}\neq \tilde{X}_{B_j}\mid X_{B_i},i<j} \\ =\sup_{E\subset A^{B_j}}|\Pr{X_{B_j}\in E\mid X_{B_i},i<j} - \Pr{X_{B_j}\in E}|.\end{array}$$
The $b$ separation condition implies that $X_{B_j}$ is $b$ steps ahead into the future from $X_{B_i},i<j$. Therefore, the $\beta$-mixing condition implies:
$$\begin{array}{l}\Pr{X_{B_j}\neq \tilde{X}_{B_j}} \\= \Ex{\sup_{E\subset A^{B_j}}|\Pr{X_{B_j}\in E\mid X_{B_i},i<j} - \Pr{X_{B_j}\in E}|}\leq \beta(b).\end{array}$$
Now observe that in order for $E(\lambda)$ to hold it suffices that $\tilde{E}(\lambda/2)$ holds and that
$$\forall w\in S,\, \frac{|N(w)-\tilde{N}(w)|}{n_w}\leq \frac{\lambda\,\pi_S}{2}.$$
Therefore we will be done once we show that:
\begin{equation}\label{eq:eventwewant}\Pr{\forall w\in S,\, \frac{|N(w)-\tilde{N}(w)|}{n_w}\leq \frac{\lambda\,\pi_S}{2}}\geq 1 - \frac{2\beta(b)}{\lambda\,\pi_S}.\end{equation}
To do this, notice that for any $w$:
$$|N(w)-\tilde{N}(w)|\leq \sum_{i=1}^r\,(|B_i|-|w|+1)_+\,\Ind{\{X_{B_j}\neq \tilde{X}_{B_j}\}}.$$
This is because each block $B_i$ may contain at most $|B_i|-|w|+1$ occurrences of $w$.

The first $r-1$ blocks have the same size $|B_i|=t$, whereas the last one cannot be larger, hence $n_w\geq (r-1)(t-|w|+1)_+$. Moreover, $X_{B_1}=\tilde{X}_{B_1}$ always. We deduce:
$$|N(w)-\tilde{N}(w)|\leq (t-|w|+1)_+\,\sum_{i=2}^r\Ind{\{X_{B_j}\neq \tilde{X}_{B_j}\}}\leq \frac{n_w}{r-1}\,\sum_{i=2}^r\Ind{\{X_{B_j}\neq \tilde{X}_{B_j}\}},$$
and \begin{eqnarray*}\Ex{\max_{w\in S}\frac{|N(w)-\tilde{N}(w)|}{n_w}}&\leq& \Ex{\max_{w\in S}\frac{|N(w)-\tilde{N}(w)|}{n_w}}\\ &\leq &\frac{\sum_{i=2}^r\Pr{X_{B_j}\neq \tilde{X}_{B_j}}}{r-1}\\ &\leq &\beta(b).\end{eqnarray*}

We deduce from Markov's inequality that:
\begin{equation*}\Pr{\max_{w\in S}\frac{|N(w)-\tilde{N}(w)|}{n_w}>\frac{\lambda\,\pi_S}{2}}\leq \frac{2\beta(b)}{\lambda\,\pi_S},\end{equation*}
and this is precisely (\ref{eq:eventwewant}), which finishes the end of the proof of the Claim.

We must now bound $\Pr{\tilde{E}(\lambda/2)}$. By the union bound, we have:
\begin{equation}\label{eq:decomposeunionbound}\Pr{\tilde{E}(\lambda/2)}\geq 1 - \sum_{w\in S}\Pr{\left|\frac{\tilde{N}(w)}{n_w} - \pi(w)\right|>\frac{\lambda\,\pi(w)}{2}}.\end{equation}
Fix a $w\in S$. Let $\tilde{N}_i(w)$ denote the number of occurrences of $w$ in $B_i$. We will apply Bennett's inequality to the sum of these random variables. To this end we note that:
\begin{enumerate}
\item $\sum_{i=1}^r\tilde{N}_i(w)=\tilde{N}(w)$.
\item {\em The $\tilde{N}_i(w)$ are independent.} This is so because the $\tilde{X}_{B_i}$ are independent.
\item {\em $\tilde{N}_i(w)\leq (|B_i|-|w|+1)_+\leq t$ for all $i$} because $t$ is an upper bound on $|B_i|$.
\item {\em $\sum_{i}\Ex{\tilde{N}_i(w)} = \pi(w)n_w$.}
\item {\em $\sum_{i}\Var{\tilde{N}_i(w)}\leq \pi(w)\,n_w\,t$.} This is so because each $\tilde{N}_i(w)$ is a sum of $(|B_i|-|w|+1)_+\leq t$ indicators with variance $\pi(w)(1-\pi(w))$ and the variance of a sum of $\leq t$ terms is at most $t$ times the sum of the variances (by Cauchy Schwarz).
\end{enumerate}
Therefore,
$$\Pr{|\tilde{N}(w)-\pi(w)\,n_w|\geq \frac{\lambda\,\pi(w)\,n_w}{2}}\leq 2\,\exp\left(-\frac{\lambda^2\,\pi(w)\,n_w}{8t(1 + \frac{\lambda}{6})}\right).$$
Since $t\geq 2h$,
$$\forall w\in S,\, n_w = \sum_{i=1}^r(|B_i|-|w|+1)_+\geq (r-1)\,(t-|w|+1)\geq \frac{(r-1)\,t}{2},$$
and the result follows from plugging the probability inequality into (\ref{eq:decomposeunionbound}) and applying the Claim.\end{proof}

From now on we set:

\begin{eqnarray}\label{eq:defbbeta}b &\equiv& \left\lceil\frac{10h_S}{\xi}\right\rceil\vee \beta^{-1}\left(\frac{\xi\,\pi_S\,\delta_0}{24}\right),\\
\label{eq:defrbeta}r&\equiv & \left\lceil \frac{n}{2b}\right\rceil.\end{eqnarray}

We now construct three sets of $b$-separated blocks in $[n]$. The first one is:
\begin{enumerate}
\item $B^{(1)}=(B_1^{(1)},\dots,B_r^{(1)})$ consists of intervals of the form
$$B_i^{(1)}\equiv\{b(2i-2)+s\,:\, 1\leq s\leq b\}\cap [n].$$
These are the intervals of length $b$ whose right endpoints are even multiples of $b$.
These are $b$-separated $b$-regular blocks.
\item $B^{(2)}=(B_1^{(2)},\dots,B_r^{(2)})$ consists of intervals of the form
$$B_i^{(1)}\equiv\{b(2i-1)+s\,:\, 1\leq s\leq b\}\cap [n].$$
These are the intervals whose right endpoints are odd multiples of $b$. In this case we set $B^{(2)}_i(w)=B^{(2)}_i$ for each $1\leq i\leq r$ and $w\in S$. This also results in $b$-separated $b$-regular blocks.
\item $B^{(3)}=(B^{(3)}_1,\dots,B^{(3)}_{2r-1})$ consists of intervals
$$B^{(3)}_i =\{bi-h_S+2,bi-h_S+3,\dots,bi+h_S\}\cap [n].$$
This results in $b$-separated $2h_S$-regular blocks, as one can check. (Here one must use $b\geq 2h_S$, which follows from (\ref{eq:defbbeta}) and the assumption $\xi<1/2$.)
\end{enumerate}
For each $k\in \{1,2,3\}$, let $N^{(k)}(w)$ count the number of occurrences of $w$ that are contained in a block of the form $B^{(k)}_i$ and let $n^{(k)}_w = \sum_i(|B_i^{(k)}|-|w|+1)_+$. We will need two propositions.
\begin{proposition}\label{claim:countinblocks}For any $w\in S$,
$$N^{(1)}(w)+N^{(2)}(w)\leq N_n(w)\leq N^{(1)}(w)+N^{(2)}(w)+N^{(3)}(w).$$\end{proposition}
\begin{proof}The LHS counts the number of occurrences of $w$ contained in intervals of the form $\{bi+1,bi+2,\dots,b(i+1)\}$. Since $|w|\leq h_S$, $N^{(3)}(w)$ is an upper bound on the number of occurrences of $w$ that are not entirely contained in one of those intervals.\end{proof}

\begin{proposition}\label{prop:numbers}For any $w\in S$, $n^{(1)}_w,n^{(2)}_w\geq \frac{(1-\xi/3)n}{2}$, $n^{(1)}_w+n^{(2)}_w\leq n$ and $n^{(3)}_w\leq \frac{3nh_S}{b}$.\end{proposition}
\begin{proof}Since $B^{(1)}$ is $b$-regular and $r\geq n/2b$ (by \ref{eq:defrbeta}),
$$n^{(1)}_w\geq (r-1)(b-|w|)\geq \frac{n}{2} - \frac{n|w|}{2b} -b = \frac{n}{2}\left(1 - \frac{h_S}{b}-\frac{b}{n}\right).$$
By (\ref{eq:defbbeta}) $b\geq 6h_S/\xi$ and $n\geq 6b/\xi$, so $n_1^{(w)}\geq (1-\xi/3)\,n/2$ as desired. The same argument works for $n^{(2)}_w$.
For $n^{(3)}_w$ we start from $2r-1\leq 2(n/2b+1)-1=n/2b+1$. Since each block contains at most $2h_S$ points,
$$n^{(3)}_w\leq \frac{nh_S}{b} + 2h_S.$$
The rest follows from $b/n\leq \xi<1$.\end{proof}

Consider the events:
\begin{eqnarray}G^{(1)}&\equiv & \left\{\forall w\in S,\, \left|\frac{N^{(1)}(w)}{n^{(1)}_w}-\pi(w)\right|<\frac{\xi\,\pi(w)}{3}\right\};\\
G^{(2)}&\equiv & \left\{\forall w\in S,\, \left|\frac{N^{(2)}(w)}{n^{(2)}_w}-\pi(w)\right|<\frac{\xi\,\pi(w)}{3}\right\};\\
G^{(3)}&\equiv & \left\{\forall w\in S,\, \frac{N^{(3)}(w)}{n}\leq \frac{2\xi\pi(w)}{3}\right\};\\
G&=& G^{(1)}\cap G^{(2)}\cap G^{(3)}.\end{eqnarray}

\begin{claim}$G\subset \{\forall w\in S\,:\,|N_n(w)-\pi(w)n|\leq \xi\pi(w)n\}.$\end{claim}
\begin{proof}Assume $G$ holds. Then for any $w\in S$:
\begin{eqnarray*}N_n(w)&\geq & N^{(1)}(w)+N^{(2)}(w)\\ \mbox{($G$ occurs)} &\geq &\pi(w)\,\left(1-\frac{\xi}{3}\right)(n^{(1)}_w+n^{(2)}_w)\\ (\mbox{Proposition} \  \ref{prop:numbers}) &\geq & \pi(w)\,\left(1-\frac{\xi}{3}\right)^2\,n\\ &\geq &\left(1-\xi\right)\,\pi(w)\,n.\end{eqnarray*}
On the other hand,
\begin{eqnarray*}N_n(w)&\leq & N^{(1)}(w)+N^{(2)}(w)+N^{(3)}(w)\\ \mbox{($G$ occurs)} &\leq &\pi(w)\,\left(1+\frac{\xi}{3}\right)\,n + \frac{2\xi\pi(w)}{4}\,n\\ &\leq &\left(1+\xi\right)\,\pi(w)\,n.\end{eqnarray*}
\end{proof}

The claim implies that:
$$\Pr{G}\leq \Pr{\forall w\in S\,:\,|N_n(w)-\pi(w)n|\leq \xi\pi(w)n}$$
and we proceed to bound $\Pr{G}$.

We first apply Lemma \ref{lem:blocks} to each $G^{(k)}$. For $k=1,2$ we may take $t=b$,
$\lambda=\xi/3$ and note that $r\geq n/2b$, $\beta(b)\leq \lambda\delta_0\pi_S/8$ (cf. \ref{eq:defrbeta}, \ref{eq:defbbeta}) to deduce:
\begin{equation}\label{eq:boundG1G2}\Pr{(G^{(1)})^c} + \Pr{(G^{(2)})^c}\leq 4|S|\,\exp\left(-\frac{\xi^2\,\pi_S\,(r-1)}{144\,\left(1 + \frac{\xi}{18}\right)}\right) + \frac{\delta_0}{4}.\end{equation}

For $G^{(3)}$ we take $t=2h_S$ and
$$\lambda =  \frac{2\xi\,n}{3\,\max_{w\in S}n_w}-1\geq \frac{2\xi\,b}{9h_S}-1\geq \frac{\xi\,b}{9h_S}$$
since $b\geq 10h_S/\xi$ by (\ref{eq:defbbeta}) and $n\geq 2b$. Notice that in this case $\lambda>1$, hence
$$\lambda^2/(1+\lambda/6)\geq \lambda^2/(\lambda +\lambda/6)\geq 3\lambda/4\geq \frac{\xi\,b}{12h_S}\geq \frac{10}{12}.$$ Moreover, $\beta(b)/\pi_S\lambda\leq \delta_0/6$. Hence, we deduce:
\begin{equation}\label{eq:boundG3}\begin{array}{rl}
\Pr{(G^{(3)})^c} & \leq 2|S|\,\,\exp\left(-\frac{\,\pi_S\,(2r-1)}{16}\,\frac{\lambda^2}{1+\lambda/6}\right)+\frac{\delta_0}{6} \\
& \leq 2|S|\exp\left(-\frac{\,\pi_S\,(2r-1)}{16}\frac{10}{12}\right)+\frac{\delta_0}{6}.\end{array}\end{equation}
Now compare the exponential terms in the two equations. Since $\xi\leq 1/2$,
$$\frac{150}{\xi^2}\geq \frac{144\,\left(1 + \frac{\xi}{18}\right)}{\xi^2}\geq 16\frac{12}{10},$$
hence the exponential term in (\ref{eq:boundG3}) is larger than the exponential term in (\ref{eq:boundG1G2}). We conclude that:
$$\Pr{G}\leq 1 - \sum_{k=1}^3\Pr{(G^{(k)})^c}\geq 1 - \frac{\delta_0}{2} - 6|S|\,\exp\left(-\frac{r-1}{\frac{300}{\xi^2\,\pi_S}}\right).$$
To finish the proof, we recall that $r\geq n/2b$ (cf. \ref{eq:defrbeta}) and notice that our assumptions imply:
$$\frac{n}{2b}\geq 1 + \frac{300}{\xi^2\,\pi_S}\ln\left(\frac{12|S|}{\delta_0}\right).$$
Plugging this back into the previous inequality gives $\Pr{G}\geq 1- \delta_0$ and finishes the proof.\end{proof}

\section{A remark on minimax rates for chains with infinite connections} \label{rem:minimaxcomplete}

In this section we observe that the uniform convergence rate obtained in Theorem \ref{thm:complete} in Section \ref{sec:complete} of the main text is optimal. We take $\ELL=1$, $A=\{0,1\}$ and omit $\ell$ from the notation.

In order to prove the optimality of our rate, it suffices to show that one can couple two processes $(X_m)_{m\in\Z}$, $(Y_m)_{n\in\N}$, with respective transition probabilities $p_X,p_Y$, so that $\Pr{X_i=Y_i,\,1\leq i\leq n} = 1-\liloh{1}$, even though $d_{\rm tv}(p_X(\cdot|\bar x,p_Y(\cdot|\bar x)) \geq \Gamma/\lceil C\log n\rceil^{1+\theta}$ for at least one past $\bar x$. Here $C>0$ is a constant that depends only on $\Gamma,\theta$ and $n\geq n_0(\Gamma,\theta)$ is assumed large enough. The upshot of this coupling construction is that no estimator can estimate the transition probabilities $p_X,p_Y$ uniformly over pasts, from samples $X_1^n$, $Y_1^n$ without making an error of magnitude at least $ \Gamma/(2\lceil C\log n\rceil^{1+\theta})$ for at least one of $p_X,p_Y$.

Our construction is as follows. First let $(X_m)_{m\in\Z}$ consist of i.i.d. uniform symbols in $A=\{0,1\}$. $Y$ is defined to have transition probabilities:
\[p_Y(1|x) \equiv \left\{\begin{array}{ll}\frac{1}{2} + \frac{\Gamma}{\lceil C\log\,n\rceil^{1+\theta}}, & x^{-1}_{-\lceil C\log n\rceil}=0\dots 0;\\ \frac{1}{2}, & \mbox{otherwise.}\end{array}\right.\]
for some $C>0$ such that $(2/3)^{C\log n}\leq 1/n^2$. Here we implictly are assuming that $n$ is large enough so that the above recipe gives valid transition probabilities. Clearly, Assumption \ref{Ex:Infinite} is satisfied by both $X$ and $Y$ when $n$ is large enough with $\eta = 2/3$. Moreover, $p_X(\cdot|\bar x),p_Y(\cdot|\bar x)$ are at distance $\Gamma/\lceil C\log n\rceil^{1+\theta}$ from each other when $\bar x$ is such that $\bar x^{-1}_{-\lceil C\log n\rceil}=0\dots 0$, i.e. with a suffix of $\lceil C\log n\rceil$ zeros.

Next we couple the two processes so that $X_1^n=Y_1^n$ with high probability. To do this, we will apply the perfect simulation algorithm of Comets et al. \cite{CometsFernandezFerrari2002}. Letting $\eps\equiv {2\Gamma}/{\lceil C\log\,n\rceil^{1+\theta}}$, first note we can write:
\[p_Y(1|x) = (1 - \eps)\,\frac{1}{2} + \eps\,q(1|x)\]
where
\[q(1|x) \equiv \left\{\begin{array}{ll}0, & x^{-1}_{-\lceil C\log n\rceil}=0\dots 0;\\ \frac{1}{2}, & \mbox{otherwise.}\end{array}\right.\]
Following the algorithm of Comets et al. \cite{CometsFernandezFerrari2002}, we may sample $Y_1,\dots,Y_n$ as follows.
\begin{enumerate}
\item[Step 1] Let $(L_m)_{m\in\N}$ be an i.i.d. sequence independent from $(X_m)_{m\in\N}$ with $\Pr{L_0=\lceil C\log n\rceil} = 1 - \Pr{L_0=0}=\eps$.
\item[Step 2] Find $T:=\min\{k\in\Z\,:\, \forall k\leq m\leq n,\, m-L_m\geq k\}$. ($T>-\infty$ a.s. by the results of \cite{CometsFernandezFerrari2002}.)
\item[Step 3] {\bf For} $m=T$ to $n$, set
\[Y_m\equiv \left\{\begin{array}{ll}0, & L_m=\lceil C\log n\rceil\mbox{ and }Y^{m-1}_{m-\lceil C\log n\rceil}=0\dots 0;\\ X_m, & \mbox{otherwise.}\end{array}\right.\]
\end{enumerate}
Notice that this defines $Y_m$ for $T\leq m\leq n$, in particular we obtain $Y_1,\dots,Y_n$. The key point observed by Comets et al. \cite{CometsFernandezFerrari2002} is that at no point of this construction we need to compute $Y_m$ for $m<T$. Indeed, for $m\geq T$, $L_m=\lceil C \log n\rceil$ implies $m\geq T + \lceil C\log n\rceil$ (by definition of $T$), which means that the values $Y^{m-1}_{m-\lceil C\log n\rceil}$ have been defined in previous executions of Step 3 (the {\bf For}  instruction). Comets et al. \cite{CometsFernandezFerrari2002} show that $Y_1^n$ is a perfect sample from the unique probability measure over $A^{\Z}$ that is compatible with $p_Y$.

Let us now show that $X^n_1=Y^n_1$ with high probability. First note that $\Pr{T<-\lceil C\log n\rceil}=\liloh{1}$. Indeed, this probability is upper bounded by $$\Pr{\bigcup_{m=-\lceil C\log n\rceil}^{-1}\{L_m\neq 0\}}\leq \lceil C\log n\rceil\,\eps = \bigoh{\frac{1}{\log^{\theta}n}}.$$
Now consider the case that $T\geq -\lceil C\log n\rceil$. If there exists an index $0\leq i\leq n$, we can see that, by taking the {\em first} such $i$, we must have $X^{i-1}_{i-\lceil C\log n\rceil}=0\dots 0$. We conclude
$$\Pr{\{T\geq -\lceil C\log n\rceil\}\cap\{X_1^n\neq Y_1^n\}}\leq \sum_{i=0}^{n}\Pr{\bigcap_{j=i-\lceil C\log n\rceil}^{i}\{X_j=0\}}.$$
Now the definition of our transitions implies that, for large $n$, \[\Pr{X_j=0\mid X_{-\infty}^{j-1}}\leq \eta=2/3,\] so $\Pr{\bigcap_{j=i-\lceil C\log n\rceil}^{i}\{X_j=0\}}\leq (2/3)^{C\log n}\leq n^{-2}$ by our choice of $C$. We deduce:
$$\Pr{\{T\geq -\lceil C\log n\rceil\}\cap\{X_1^n\neq Y_1^n\}}\leq (n+1)n^{-2}=\liloh{1}.$$
Since we have already shown $\Pr{T< -\lceil C\log n\rceil}=\liloh{1}$, we see that $\{X_1^n\neq Y_1^n\}$ has vanishing probability, as desired.

\bibliography{biblioVLMC}
\bibliographystyle{plain}

\end{document}